\documentclass[11pt]{article}
\usepackage{amsmath,amssymb,amsthm,graphicx,setspace,bigstrut}
\usepackage[margin=1in]{geometry}
\usepackage[normalem]{ulem}
\newcommand{\stkout}[1]{\ifmmode\text{\sout{\ensuremath{#1}}}\else\sout{#1}\fi}
\allowdisplaybreaks
\usepackage[colorlinks,urlcolor = blue,citecolor=blue]{hyperref}
\usepackage[citestyle=authoryear,hyperref=true,sortcites=false,uniquename=false,mincitenames=1,maxcitenames =4,uniquelist=false,maxbibnames=10,bibstyle=authoryear,giveninits=true]{biblatex}
\renewbibmacro{in:}{}
\DeclareNameAlias{sortname}{family-given}
\newtheorem{lemma}{Lemma}

\newtheorem{theorem}{Theorem}
\newtheorem{assumption}{Assumption}
\DeclareMathOperator*{\argmin}{arg\,min}
\DeclareMathOperator*{\argmax}{arg\,max}
\begin{document}
\title{Faster estimation of dynamic discrete choice models\\ using index invertibility}
\author{Jackson Bunting\\Department of Economics\\ University of Washington\thanks{(Corresponding author) Department of Economics, the University of Washington, 1410 NE Campus Pkwy, Seattle, WA 98195. Email address: \href{mailto:buntingj@uw.edu}{buntingj@uw.edu}.} \and Takuya Ura\\Department of Economics\\ University of California, Davis\thanks{Department of Economics, University of California Davis, 1 Shields Ave. Davis, CA 95616. Email address: \href{mailto:takura@ucdavis.edu}{takura@ucdavis.edu}.}}
\maketitle 

\begin{abstract} 
Many estimators of dynamic discrete choice models with persistent unobserved heterogeneity have desirable statistical properties but are computationally intensive. In this paper we propose a method to quicken estimation for a broad class of dynamic discrete choice problems by exploiting semiparametric index restrictions. Specifically, we propose an estimator for models whose reduced form parameters are invertible functions of one or more linear indices \parencite{ahn2018simple}, a property we term index invertibility. We establish that index invertibility implies a set of equality constraints on the model parameters. Our proposed estimator uses the equality constraints to decrease the dimension of the optimization problem, thereby generating computational gains. Our main result shows that the proposed estimator is asymptotically equivalent to the unconstrained, computationally heavy estimator. In addition, we provide a series of results on the number of independent index restrictions on the model parameters, providing theoretical guidance on the extent of computational gains. Finally, we demonstrate the advantages of our approach via Monte Carlo simulations.
\begin{description}
\item[Keywords:]  Dynamic discrete choice, multiple-index model, pairwise differences, semiparametric regression.
\item[JEL Codes:] C01, C63
\end{description}
\end{abstract}

\newpage

\section{Introduction}

In dynamic discrete choice modeling, estimation of the structural parameters that underlie economic decisions is often computationally challenging. Many available estimators for the structural parameter of interest $\theta_{0}\in\Theta$ are extremum estimators: 
\begin{equation}\label{eq:inf_optimization}
\hat{\theta}^*=\arg\max_{\theta\in\Theta}\hat{Q}(\theta).
\end{equation}
For instance, the criterion function $\hat{Q}$ may be the log-likelihood function \parencite{john1988maximum}, a pseudo log-likelihood function \parencite{hotz1993conditional,arcidiacono2011conditional} or a minimum distance function \parencite{pesendorfer2008asymptotic}.  While these estimators offer appealing theoretical properties, they often impose substantial computational demands for multiple reasons. First, evaluating the criterion function may involve solving the model through costly fixed-point iteration or by simulation. Second, the criterion function's global concavity is not always guaranteed, often necessitating the use of global optimization methods or initializing a local optimization algorithm at various starting values. A relevant case is finite mixture models whose likelihood function may lack global concavity \parencites[e.g.,][ p. 182]{robert1999monte}[]{arcidiacono2011conditional}.

In this paper we harness the index restrictions inherent in many structural models to introduce an estimator for $\theta_0$ that offers substantial computational advantages and is asymptotically equivalent (of arbitrarily high order) to $\hat{\theta}^*$. Our focus is on models satisfying a condition we term `index invertibility'. Drawing from the semiparametric index regression literature, we describe a model as index invertible if its reduced form parameters are an invertible function of a vector of linear indices \parencite{ahn2018simple}. We establish that index invertibility implies a set of linear equality constraints which constrain $\theta_0$ to belong in a subspace of $\Theta$,\footnote{The subspace may be a strict subspace of $\Theta$ when there is at least one continuous covariate. We conjecture that it is possible to extend our method with inequality constraints when there is no continuous covariate \parencite{khan2018discussion}.} thereby reducing the dimension of the optimization problem presented in equation \eqref{eq:inf_optimization}. The main contribution of our paper is to propose an estimator which implements the constraints implied by index invertibility, and prove its asymptotic equivalence to the computationally intensive estimator $\hat{\theta}^*$.

Arguably, the class of index invertible structural econometric models is very broad. First, we prove that a broad class of dynamic discrete choice models with persistent unobserved heterogeneity (i.e., unobserved state variables) and partially linear flow utility satisfy index invertibility. In this leading example of index invertibility, the reduced form parameters are the conditional choice probabilities (defined as the probability of each choice conditional upon the observed state variables) which depend on multiple indices which govern the flow utility and transition of the state variables. Second, we do not restrict nor require specification of the number of indices required to attain index invertibility. Of course, as we show formally, the computational gains of our approach may diminish as the number of indices required to achieve index invertibility grows. Finally, the condition encompasses many invertible index models in the literature (see, for example, \textcite{ahn2018simple} and references therein).

Our approach is based on the observation that index invertibility implies a set of equality constraints on the structural parameter. Namely, we show that under index invertibility, the true parameter value $\theta_0$ satisfies $$\Sigma_{0}\boldsymbol{\gamma}(\theta_0)=0$$ for a known linear function $\boldsymbol{\gamma}(\cdot)$ and a nonparametrically identified matrix ${\Sigma}_0$. If $\Sigma_0$ were known, to solve the population version of equation \eqref{eq:inf_optimization} it would be sufficient to search among $\boldsymbol{\gamma}(\theta)$ in the nullspace of $\Sigma_0$. Our estimator builds upon this idea and is defined by the following two steps: first, given an estimator $\hat{\Sigma}$ for $\Sigma_{0}$ (e.g., kernel smoothing in Section \ref{sec:const_tildeSigma}) we compute
\begin{equation}\label{eq:tildetheta}
    \tilde{\theta}=\argmax_{\theta\colon\hat{\Sigma}\boldsymbol{\gamma}(\theta)=0}\hat{Q}\left(\theta\right).
\end{equation}
Solving the optimization problem in equation \eqref{eq:tildetheta} is computationally simpler than the unconstrained problem in equation \eqref{eq:inf_optimization} as it is only necessary to search over parameter values in $\{\theta\in\Theta\colon\hat{\Sigma}\boldsymbol{\gamma}(\theta)=0\}\subseteq\Theta$.\footnote{Since the constraints are linear in parameters, once estimated, they can be imposed generically at negligible computational cost. See Section \ref{app:mcmcsigma} for a general description.} The second step is to apply Newton-Raphson updates from $\tilde{\theta}$ in the direction of the target estimator $\hat\theta^*=\arg\max_{\theta\in\Theta}\hat{Q}(\theta)$ \parencite{robinson1988stochastic}. The resulting estimator is asymptotically equivalent to the more computationally intensive $\hat{\theta}^*$. Notably, given the typical statistical justification for $\hat{\theta}^*$ relies on asymptotic approximations (of a certain order), our proposed estimator inherits the favorable statistical properties of $\hat{\theta}^*$ but is computationally more efficient. To illustrate, if $\hat{\theta}^*$ stands for the parametric maximum likelihood estimator, then, under standard conditions, our method can achieve the Cram\'{e}r-Rao bound at lower computational cost by leveraging semiparametric index restrictions. Regardless, we emphasize that, under the conditions elaborated in the following sections, our method can be used to target any extremum estimator $\hat\theta^*$ within a model satisfying index invertibility---for example, if $\hat\theta^*$ is motivated by computational considerations, our method may be used to further reduce the computational burden.

As computational efficiency motivates our estimator, it is natural to explore the magnitude of possible computational benefits. Section \ref{sec:rank_Sigma} provides some theoretical insights on this question. Recall that the computational gains arise from imposing the constraints $\Sigma_0\boldsymbol{\gamma}(\theta_0)=0$. Thus a key determinant of the computational benefits of our estimator is the rank of $\Sigma_0$: the larger the rank of $\Sigma_0$, the more restrictions $\Sigma_0\boldsymbol{\gamma}(\theta_0)=0$ places on $\theta_0$. Using the definition of $\Sigma_0$ (equation \eqref{eq:sigma0}), we develop a series of results on the rank of $\Sigma_0$. Our results suggest two situations where the computational gains of our method will be large: either in the presence of many continuous covariates, or in the presence of at least one continuous covariate that satisfies a particular rectangular support condition. Our results also suggest that the rank of $\Sigma_0$ may decrease with the number of indices required to attain index invertibility.

To illustrate the advantages of our method, we consider a Monte Carlo experiment based on the utility function specification of \textcite{toivanen2005market}.\footnote{We consider a single-agent model instead of the two-agent model in \textcite{toivanen2005market}.} They estimate a dynamic model of firm entry into the U.K. fast food market between 1991 and 1995. In their model, firm profits from entry are determined by market size, which is modeled as depending on a long vector of socio-economic variables \parencite[e.g.,][]{bresnahan1991entry,toivanen2005market,aguirregabiria2020identification}. Due to the availability of these continuous socio-economic variables, our method is able to feasibly apply 8 restrictions to the 11-dimensional payoff parameter vector. Specifically, we consider two standard estimators as target estimators for our method: one based on the approach of \textcite{bajari2011simple}, the other on \textcite{arcidiacono2011conditional}. By simulating data from this model, we demonstrate that our estimator is, on average, around 20 times computationally more efficient than the standard approaches to estimating the model,\footnote{Our simulations demonstrate two benefits of effectively reducing the dimension of an optimization problem that requires repeated initializations: first, fewer initializations are required to reasonably cover the lower dimensional space; second, the lower dimensional problem is easier to solve \textit{per} initialization. In our simulations, roughly, we use three times fewer initializations for the lower dimensional problem, and our estimator is around six times faster \textit{per} initialization. Of course, while the number of required initializations may increase exponentially with the dimension of the optimization problem, it remains at the user's discretion and may be parallelized.} and provide empirical validation of our main theoretical result.

Our proposed method aims to contribute to a large literature on the computational aspects of structural modeling and, in particular, dynamic discrete choice \parencites[e.g., ][]{hotz1994simulation,arcidiacono2011conditional,bajari2011simple,su2012constrained,arcidiacono2013approximating,kristensen2021solving}. Rather than proposing an alternative to computationally advantageous estimators in the literature, our method can be used to improve the computation time for any estimator that can be expressed as the maximizer of a smooth sample criterion function. For instance, in our Monte Carlo experiment, we apply our method to estimators based on \textcites{arcidiacono2011conditional,bajari2011simple}, the latter of which is especially known to be computationally attractive.  Parts of this paper are closely related to \textcite{ahn2018simple}, who develop a computationally simple estimator for a class of invertible index models. Whereas their paper focuses on identification and estimation of the index parameter, we allow the index parameter to be one part of a broader structural model and harness the semiparametric index restrictions for computational purposes.

The remainder of the paper is structured as follows. To fix ideas on our leading example of index invertibility, Section \ref{sec:ddc} considers the specific context of dynamic discrete choice models. We then formally introduce the general model and index invertibility, and derive the equality constraints implied by index invertibility (Section \ref{sec:model}). Section \ref{sec:rank_Sigma} derives bounds on the rank of $\Sigma_0$, an important determinant of the number of independent restrictions in $\Sigma_0\boldsymbol{\gamma}(\theta_0)=0$. Section \ref{sec:estimation} outlines the estimator and derives its equivalence to the computationally intensive estimator, our main result. Finally, Section \ref{sec:appl} presents our Monte Carlo experiment. In the Appendices we gather proofs, a proposed consistent estimator for $\Sigma_0$, and additional details on the Monte Carlo exercise.

\section{Index invertibility in dynamic discrete choice models}\label{sec:ddc}

In this section, we introduce the idea of `index invertibility' within the context of dynamic discrete choice (DDC) models. We have two goals for this section: first, to illustrate the idea of index invertibility through a simple and classical DDC model (Section \ref{sec:simple}); second, to introduce a broad, empirically relevant class of DDC models and prove that they satisfy index invertibility (Section \ref{subsec:ddcm}). In Section \ref{sec:model}, we discuss index invertibility in a general model. We emphasize that the model in this section is just \textit{one} example of sufficient conditions for index invertibility within a DDC model: the conditions specified in this section do not preclude other DDC models from possessing the index invertibility property (cf. Section 5).

\subsection{An illustrative dynamic discrete binary choice model}\label{sec:simple}

In each period $t=1,2,\dots,T=\infty$, an agent observes a state variable $S_{t}$ and chooses an action $A_{t}\in\mathcal{A}=\{0,1\}$ to maximize their expected discounted utility.\footnote{The result of this section applies to $T<\infty$ (i.e., a non-stationary problem). We present only the $T=\infty$ case for notational ease.} The state variable is composed of two subvectors, $Z_{t}$ and $\epsilon_{t}$ which are observed and unobserved to the econometrician, respectively. We suppose the distribution of $\epsilon_t\in\mathbb{R}$ is known and has full support. The agent has time-separable utility and discounts future payoffs by the known rate $\beta_0\in[0,1)$.\footnote{Here and throughout the paper we use the subscript $0$ to indicate the true parameter value.} The period $t$ payoff for action $1$ is given by $Z_t^\intercal\gamma_0+\epsilon_t$, where $\gamma_0\in\mathbb{R}^{\dim(Z)}$ is unknown. The flow payoff to action $A_t=0$ is 0. 

We assume $\{Z_t,\epsilon_t,A_t\}$ follows a stationary first-order Markov process and satisfies the following conditional independence assumption:
\begin{equation*}
d\Pr(Z_{t+1},\epsilon_{t+1}\mid Z_t,\epsilon_t,A_t)=dF_Z(Z_{t+1}|Z_t,A_t)\times dF_{\epsilon}(\epsilon_{t+1}).
\end{equation*}
We further impose  $$F_Z(z'|z,a)=G(z',\delta_0^\intercal z,a)$$ for some function $G$ and $\delta_0\in\mathbb{R}^{\dim(Z)\times J_2}$. This assumption is substantive when $J_2<\dim(Z)$,\footnote{The existence of $(\delta_0,G)$ is without loss of generality since one can always set $\delta_0$ equal to the identity matrix and $G=F_Z$ (in this case, $J_2=\dim(Z)$). As we explain below, the computational gain from our method comes from the rank of the matrix $\Sigma_0$ in equation \eqref{eq:sigma0}. In the trivial case of $\delta_0$ equaling identity, the rank of $\Sigma_0$ is zero, and our method does not provide a computational gain.} and is testable from observed data since $F_Z$ is nonparametrically identifiable.\footnote{For Lemma \ref{theorem:A1_simple_model}, it is sufficient that the expected future utility $E[{v}(Z_{t+1})\mid Z_t=z,A_t=a]$ depends on $z$ only through $\delta_0^\intercal z$.} Since $G$ and $\delta_0$ are nonparametrically identified, they can be consistently estimated in a computationally feasible manner. Notice that the above conditions allow the special case that the transition of the state variable is deterministic (e.g., the lagged choice is a state variable), but in general the state transition may be unknown but nonparametrically identified.

Now define the conditional choice probability function as 
$$
\Pi_0(z)=\Pr(A_{t}=1\mid{Z}_{t}=z)=\Pr\left(v(1,z)+\epsilon_{t}>v(0,z)\right),
$$
where $v(a,z) = \gamma_0^\intercal z 1\{a=1\}+\beta E[{v}(Z_{t+1})\mid Z_t=z,A_t=a]$ and $v(z)$ is the equilibrium ex-ante value function.\footnote{The ex-ante value function is defined as the discounted sum of future payoffs from optimal behavior given $Z_t=z$  but before the agent observes $\epsilon_t$ and chooses $A_t$, which we assume is well defined as a function of $z$ (\textcite{john1988maximum}. See, e.g., \textcite[p. 11]{aguirregabiria2007sequential}  or \textcite[p. 1036]{bugni2021iterated}.} In this simple model, the unknown structural parameter includes the payoff parameter $\gamma$, and a nuisance parameter $\delta$ which governs the state transition probabilities and is identified directly from the data. We now show that, if the state transition index $\delta_0^\intercal z$ is held fixed, the conditional choice probabilities $\Pi_0$ are an invertible function of the payoff index $\gamma_0^\intercal z$, an example of a property we term {`index invertibility'} (Assumption \ref{assn:key_equivalence}):

\begin{lemma}\label{theorem:A1_simple_model}
For the dynamic discrete choice problem in this section, 
\begin{equation}\label{eq:inv_theorem_simple}
\Pi_0(z_1)=\Pi_0(z_2)
\iff
\gamma_0^\intercal z_1=\gamma_0^\intercal z_2
\end{equation}
for every pair of points, $z_1$ and $z_2$, in the support of $Z$ with $\delta_0^\intercal z_1=\delta_0^\intercal z_2$.
\end{lemma}
\begin{proof}
From the definition of the conditional choice probability, we have 
\begin{align*}
\Pi_0(z)=\Pr\left(\epsilon_{t}>-\gamma_0^\intercal z -\beta\int{v}(z')\left(G(dz';\delta_0^\intercal z,1)-G(dz';\delta_0^\intercal z,0)\right)
\right).
\end{align*}
Since $\epsilon_t$ has full support, the function $$u_1\mapsto \Pr\left(\epsilon_{t}>-u_1 -\beta\int{v}(z')\left(G(dz';u_2 ,1)-G(dz';u_2,0)\right)
\right)$$ is injective for every $u_2$. Therefore, for every pair of points, $z_1$ and $z_2$, in the support of $Z$ with $\delta_0^\intercal z_1=\delta_0^\intercal z_2$ (i.e., $=u_2$), we have 
$$
\Pi_0(z_1)=\Pi_0(z_2)
\iff
\gamma_0^\intercal z_1=\gamma_0^\intercal z_2.~\eqno\qedhere
$$
\end{proof}

The condition in equation \eqref{eq:inv_theorem_simple} provides the basis of the computational savings of our method. To conclude the current section, we offer some intuition for this connection. To explain, suppose $\delta_0$ and $\Pi_0$ are known (note that they are identified directly from the data and can be estimated in a computationally attractive manner). If there are two values $z=z_1$ and $z=z_2$ that yield the same value of $(\delta_0^\intercal z,\Pi_{0}(z))$, then equation \eqref{eq:inv_theorem_simple} implies that the true payoff parameter $\gamma_0$ satisfies $\gamma_0^\intercal (z_1-z_2)=0$. Thus, purely from knowledge of this $z_1,z_2$, we learn that $\gamma_0$ belongs in the strict subspace of $\mathbb{R}^{\dim(\gamma)}$, $\{\gamma:\gamma^\intercal (z_1-z_2)=0\}$, and thus the effective dimension of $\gamma_0$ has decreased. In Section \ref{sec:estimation} we formally introduce our estimator which builds on this idea, broadening the number of pairs of $Z$ which are used to find the lower dimensional subspace (i.e., Theorem \ref{theorem:pair_diff}), and taking into account that $\delta_0$ and $\Pi_0$ may be estimated.

\subsection{Dynamic discrete choice models with unobserved heterogeneity}
\label{subsec:ddcm}

In this section we introduce a broad class of dynamic discrete choice models that satisfy index invertibility. The choice set is $\mathcal{A}=\{0,1,\ldots,J_1\}$, and the state variable consists of three subvectors, $Z_{t}$, $\lambda_{t}$ and $\epsilon_{t}$ where $Z_{t}$ and $(\lambda_{t},\epsilon_{t})$ are observed and unobserved to the econometrician, respectively. The unobserved components $\epsilon_{t}$ are action specific, i.e., $\epsilon_{t}\in\mathbb{R}^{J_1+1}$, and we assume $\epsilon_{t}$ has full support. The period $t$ payoff for action $a$ is given by 
\begin{equation}\label{eq:flow_gen}
\gamma_0(a)^\intercal Z_t+f(\delta_{u,0}^\intercal Z_t,a,\lambda_{t})+\epsilon_t(a),   
\end{equation}
where $\gamma_0(a)\in\mathbb{R}^{\dim(Z)}$, $\delta_{u,0}\in\mathbb{R}^{\dim(Z)\times J_{2,u}}$, and $f$ is a (possibly) nonlinear function.\footnote{We do not claim that the flow payoff is identified nonparametrically without, e.g., restrictions on $\delta_{u,0}$ and $f$. In applications of interest, $\delta_{u,0}$ may be a known matrix that governs which elements of $Z_{t}$ enter $f$, and $f$ may be a parametric function. As explained in Section \ref{sec:model}, we assume that some structural parameter of interest $\theta_0$ is point identified, and focus on harnessing the semiparametric structure to efficiently estimate the point identified parameter.} Note that, as a special case, some parts of the index parameter $\gamma(a)$ may be zero (and known), so that the corresponding elements of the observed state $Z_{t}$ affect payoffs only via the nonlinear function $f$. In other words, the flow utility function may be `partially linear' in its arguments---it allows for the case that only a subset of the state variables enter linearly.

We collect the parameter $\gamma_0(a)$ over $a\in\mathcal{A}$ in a matrix $\gamma_0=\left[\gamma_0(1),\ldots,\gamma_0(J_1)\right]\in\mathbb{R}^{\dim(Z)\times J_1}$, and impose the outside good assumption $ \gamma_0(0)=0$ and $f(\delta_{1,0}^\intercal Z_t,0,\lambda_{t})=0$. As usual, suppose $\{Z_t,\epsilon_t,\lambda_t,A_t\}$ follows a stationary first-order Markov process and satisfies 
\begin{align}
d\Pr(\epsilon_{t+1},z_{t+1},\lambda_{t+1}|\epsilon_t,z_t,\lambda_{t},a_t)&=dF_{\epsilon}(\epsilon_{t+1})\times dG(z_{t+1},\delta_{F,0}^\intercal z_t, a_t)\times dF_\lambda(\lambda_{t+1}|\lambda_{t})\label{eq:conditional_independence_lambda}
\end{align}
for some $(\delta_{F,0},G)$ such that $$F_Z(z'|z,a)=G(z',\delta_{F,0}^\intercal z, a).$$
The condition in \eqref{eq:conditional_independence_lambda} strengthens the standard conditional independence assumption of \textcite{rust1994structural} by imposing a stronger independence condition on $\epsilon_t$, and a type of conditional independence between $\lambda_t$ and $Z_t$. Within this model, we define $\Pi_{0}(z)=\{\Pi_0(a,z)\colon a = 0, 1, \ldots, J_{1}\}$ as a mixture of the $\lambda_{t}$-specific conditional choice probability function, that is
$$
\Pi_0(a,z)=\int\Pr\left(a=\arg\max_{\tilde{a}\in\mathcal{A}}\left\{v(\tilde{a},Z_t,\lambda_t)+\epsilon_{t}(\tilde{a})\right\}\mid Z_t=z,\lambda_{t}=l\right)dF_{\lambda_t}(l),
$$
where $v(a,Z_t,\lambda_t) =\gamma_0(a)^\intercal Z_t+f(\delta_{u,0}^\intercal Z_t,a,\lambda_t)+\beta E[{v}(Z_{t+1},\lambda_{t+1})\mid Z_t,\lambda_t,A_t=a]$ and $v(Z_{t+1},\lambda_{t+1})$ is the equilibrium ex-ante value function. In Appendix \ref{sec:A1_proof}, we invoke the results of \textcite{kasahara2009nonparametric} to show that $\Pi_0$ can estimated directly from the data, and prove the next result that states that the above model is index invertible.

\begin{theorem}\label{theorem:A1}
The dynamic discrete choice problem of Section \ref{subsec:ddcm} satisfies index invertibility in the sense that, for each pair $z_1$ and $z_2$ in the support of $Z$ that satisfy $\delta_0^\intercal z_1=\delta_{0}^\intercal z_2$ for $\delta=[\delta_{u}^\intercal,\delta_{F}^\intercal]^\intercal$,
\begin{equation*}
\Pi_0(z_1)=\Pi_0(z_2)
\iff
\gamma_0^\intercal z_1=\gamma_0^\intercal z_2.
\end{equation*}
\end{theorem}

Relative to the illustrative model of Section \ref{sec:simple}, the model for which we establish index invertibility in Theorem \ref{theorem:A1} is substantially richer. First, it allows for multinomial choice (e.g., occupational choice, number of operating plants, etc.). Second, the model includes persistent unobserved heterogeneity $\lambda_{t}$. In the general case that $\lambda_{t}$ is time-varying, this variable may be referred to as the (vector of) `unobserved states'; in the special case that $\lambda_{t}$ is fixed across time (i.e., $\lambda_{t}=\lambda$), it is often called the agent/market's `type' or permanent unobserved heterogeneity. Furthermore, the persistent unobserved heterogeneity may enter the flow payoff function non-linearly with states through $f$. Fourth, the flow utility function is allowed to be non-linear in states. Fifth, as noted above, as a special case the flow utility may be `partially linear'---that is, only some part of the flow utility may be linear---which appears to be a common functional form in empirical work (see, among others, \textcites{arcidiacono2005affirmative,todd2006assessing,kennan2011effect,beffy2012choosing}).\footnote{At some notational cost, we could generalize the flow utility function in equation \eqref{eq:flow_gen} to $h(\gamma_0(a)^\intercal Z_t)+f(\delta_{u,0}^\intercal Z_t,a,\lambda_{t})+\epsilon_t(a)$ where $h$ is an invertible function. For example, $h$ may be the exponential function.} Thus, anticipating the results of Section \ref{sec:rank_Sigma}, our method has the potential to generate computational gains in DDC models with `partially linear' flow utility in the presence of continuous observed state variables. 

\section{A general model with index invertibility} \label{sec:model}

In this paper we are interested in learning a finite-dimensional parameter vector $\theta_0\in\Theta$ that is identified as the unique maximum of a population criterion function ${Q}(\theta)$:   
$$
{\theta}_0=\arg\max_{\theta\in\Theta}{Q}(\theta).
$$
For instance, $\theta_0$ may be the point-identified structural parameter (sub-)vector in a dynamic discrete choice model. If $\hat{Q}$ is an estimator for $Q$, one may estimate $\theta_0$ by 
$$
\hat{\theta}^*=\arg\max_{\theta\in\Theta}\hat{Q}(\theta).
$$ 
However, in many cases, finding the maximum of $\hat{Q}(\theta)$ over $\Theta$ may be computationally challenging. For instance, if $\hat{Q}(\theta)$ represents the sample log-likelihood function of a DDC model \parencite{rust1994structural,aguirregabiria2002swapping}, $\hat{Q}(\theta)$ may lack a known closed form, requiring iterative or simulation methods to compute. Moreover, in the presence of unobserved types and/or states, $\hat{Q}(\theta)$ may lack global concavity \parencite{arcidiacono2011conditional}, necessitating repeated initialization of the optimization algorithm.

Our goal is to obtain an asymptotically equivalent estimator to  $\hat{\theta}^*$ in a computationally feasible way. We achieve this by incorporating the following restriction into the optimization. 

\begin{assumption}[Index invertibility]\label{assn:key_equivalence}
Let $\boldsymbol{\gamma}(\theta)\in\mathbb{R}^{\dim (Z) \times J_1 }$ be a known linear function of $\theta$ and $Z\in \mathbb{R}^{\dim (Z)}$ be a random vector. Denote $\gamma=\boldsymbol{\gamma}(\theta)$. There exists functions $Z\mapsto\Pi_0(Z)$ and $\delta_0\in\mathbb{R}^{\dim (Z)\times J_2}$ such that, for every pair of points, $z_1$ and $z_2$, in the support of $Z$ with $\delta_0^\intercal z_1=\delta_0^\intercal z_2$,
$$
\Pi_0(z_1)=\Pi_0(z_2)
\iff
\gamma_0^\intercal z_1=\gamma_0^\intercal z_2.
$$
\end{assumption}

We refer to Assumption \ref{assn:key_equivalence} as index invertibility. It states that for a known linear function $\boldsymbol{\gamma}(\theta)$ of the parameter of interest, the random variable $Z$ can be used to construct a vector of indices $[\gamma_0,\delta_0]^\intercal Z$ for which $\Pi_0(Z)$ is an invertible function of $\gamma_0^\intercal Z$, while $\delta_0^\intercal z$ is held fixed. It is worth noting that the qualifier $\delta_0^\intercal z_1=\delta_0^\intercal z_2$ is included to make Assumption \ref{assn:key_equivalence} apply more generally: we allow for the case that $\delta_0$ is the $ \dim (Z)\times 1$ zero vector. We can interpret $J_1+J_2$ as the number of indices required to achieve index invertibility. For example, in the DDC example of Section \ref{sec:simple}, $J_1$ is the number of index parameters that enter the flow utility (i.e., the size of choice set less the outside option), and $J_2$ is the number of indices that enter the state transition. As we show in Section \ref{sec:rank_Sigma}, the computational benefits of our approach depend importantly on $J_1$ and $J_2$. 

In practice, one's choice of $(\delta_0,\Pi_0)$ is best guided by application-specific knowledge and the chosen structural model. For example, in the model of Section 5, one part of the flow profits depends on a set of time-invariant socio-economic variables. To be more specific, although profits depend on both the (time-varying) number of operating stores and the (time-invariant) socio-economic variables, we assume that the available market size depends only on the latter. Then, as we show, the probability of operating at least one store in a given time period is increasing in the available market size. In this case, we can use $\delta_0$ to hold fixed the time period, and $\Pi_0$ is naturally related to the probability that the firm operates at least one store.

In order to exploit index invertibility, we define the matrix
\begin{equation}\label{eq:sigma0}
{\Sigma}_0\equiv E[(Z_1-Z_2)(Z_1-Z_2)^\intercal\mid \Pi_0(Z_1)=\Pi_0(Z_2),\delta_0^\intercal Z_1=\delta_0^\intercal Z_2],
\end{equation}
where $Z_1$ and $Z_2$ are independent random variables with the same marginal distribution as $Z$. The next result shows that $\Sigma_{0}$ characterizes the equality constraints that are implied by index invertibility.

\begin{theorem}\label{theorem:pair_diff}
Under Assumption \ref{assn:key_equivalence}, ${\Sigma}_0=E[(Z_1-Z_2)(Z_1-Z_2)^\intercal\mid [\gamma_0,\delta_0]^\intercal(Z_1-Z_2)=0]$ and 
\begin{equation}
{\Sigma}_0\gamma_0=0.\label{eq:1}
  \end{equation}
\end{theorem}
\begin{proof}
By Assumption \ref{assn:key_equivalence} and equation \eqref{eq:sigma0}, we have $${\Sigma}_0=E[(Z_1-Z_2)(Z_1-Z_2)^\intercal\mid [\gamma_0,\delta_0]^\intercal(Z_1-Z_2)=0].$$ Therefore, ${\Sigma}_0\gamma_0=E[(Z_1-Z_2)(Z_1-Z_2)^\intercal\gamma_0\mid [\gamma_0,\delta_0]^\intercal(Z_1-Z_2)=0]=0$.
\end{proof}

Theorem \ref{theorem:pair_diff} shows that index invertibility (Assumption \ref{assn:key_equivalence}) implies that $\theta_0$ satisfies $\dim (Z)\times J_1$ equality constraints, namely that $\theta_0\in\{\theta\in\Theta\colon{\Sigma}_0\gamma_0=0\}\subseteq\Theta$. If $\Sigma_0$ were known, then imposing the equality constraints in the optimization problem (equation \eqref{eq:inf_optimization}) necessarily eases the computational burden, since the search is limited to a smaller set of possible parameter values. Our estimator (described in Section \ref{sec:estimation}) builds on these ideas.

Equation \eqref{eq:1} suggests there are $\dim (Z)\times J_1$ restrictions on $\theta_0$, however, in practice, these restrictions may be linearly dependent. From equation \eqref{eq:1}, a key determinant of the number of linearly independent restrictions is the rank of $\Sigma_0$: in the extreme case that $\Sigma_0=0$, there are no restrictions on $\theta_0$ from $\Sigma_0\gamma_0=0$; in the other extreme case that the rank of $\Sigma_0$ is $\dim (Z)-1$, then there are $(\dim (Z)-1)\times J_1$ restrictions on $\boldsymbol{\gamma}(\theta)$ \parencite[e.g.,][]{ahn2018simple}. Given the importance of the number of linearly independent restrictions to the benefits of imposing the equality constraints, we now provide some results on the rank of $\Sigma_0$ (Section \ref{sec:rank_Sigma}).

\subsection{Rank of constraint matrix}\label{sec:rank_Sigma}

In this section, we consider the rank of $\Sigma_{0}\in\mathbb{R}^{\dim(Z)\times \dim(Z)}$, which determines the strength of restrictions implied by index invertibility. Under index invertibility, each column of the structural parameter $\gamma_0\in\mathbb{R}^{\dim(Z)\times J_1}$ belongs in the nullspace of $\Sigma_{0}$, which has dimension $\dim(Z)-\mathrm{rank}(\Sigma_0)$ by the rank-nullity theorem. Ergo, the effective number of restrictions on $\gamma_0$ implied by index invertibility is rank$(\Sigma_0)\times J_1$. That is, the larger the rank of $\Sigma_{0}$, the greater the computational advantage of imposing the equality constraints $\Sigma_0\gamma_0=0$. In this section, we provide some sufficient conditions under which a lower bound on $\mathrm{rank}(\Sigma_0)$ can be characterized.

To summarize the findings in broad terms, the results of this section provide two routes to achieving a high $\mathrm{rank}(\Sigma_0)$: either by having many continuous components of $Z$ (Theorem \ref{theorem:rank_Sigma2}), or by having one continuous component of $Z$ that satisfies a particular support condition (Theorem \ref{theorem:rank_Sigma3}). In terms of practical guidance for applied work, the results of this section can be viewed as suggestive of the type of empirical settings where our results may generate large effective dimension reduction. Namely, in an index invertible model with either many continuous covariates, or one continuous covariate that satisfies a rectangular support condition, our method may generate large computational savings. For example, in the dynamic discrete choice model studied in  Section \ref{sec:appl}, $Z$ includes 9 continuous components, and we can apply Theorem \ref{theorem:rank_Sigma2} to show $\Sigma_0$ implies 8 restrictions on the structural parameter.

The first theorem provides a lower bound on the rank of $\Sigma_0$ which depends on the number of continuous components of $Z$, but may be lower when the number of indices $J_1+J_2$ is larger.

\begin{theorem}\label{theorem:rank_Sigma2}
Suppose $Z=[Z_{A}^\intercal,Z_{B}^\intercal]^\intercal$ and there is a support point $z=[z_A^\intercal,z_B^\intercal]^\intercal$ of $Z$ such that $z_A$ is an interior point of the conditional support of $Z_{A}$ given $Z_{B}=z_B$. Then 
$$
\mathrm{rank}({\Sigma}_0) \ge \dim(Z_A)-\mathrm{rank}(Var([\gamma_0,\delta_0]^\intercal[Z_{A}^\intercal,0^\intercal]^\intercal)).
$$
Furthermore, if $\delta_0^\intercal Z$ is discrete, then
$$
\mathrm{rank}({\Sigma}_0)\geq \dim(Z_A)-J_1.
$$
\end{theorem}

By the interior-point assumption, the variable $Z_A$ is continuously distributed (given $Z_B$). The term $\mathrm{rank}(Var([\gamma_0,\delta_0]^\intercal[Z_{A}^\intercal,0^\intercal]^\intercal))$ represents how many components in $[\Pi_0(Z),Z^\intercal\delta_0]^\intercal$ are continuously distributed. It is naturally bounded above by $J_1+J_2$, the number of indices required to achieve index invertibility---Theorem \ref{theorem:rank_Sigma2} states that is preferable for this number to be small relative to the number of continuous components of $Z$. In particular, the second part of Theorem \ref{theorem:rank_Sigma2} states it is desirable for the non-structural index $\delta_0^\intercal Z$ to depend only on discrete components of $Z$. In this case $J_1$ is an upper bound for $\mathrm{rank}(Var([\gamma_0,\delta_0]^\intercal[Z_{A}^\intercal,0^\intercal]^\intercal))$. To provide a concrete example, in a dynamic discrete choice problem this would occur if the state transition depended only upon lagged actions and discrete state variables.

The second theorem states that if one component of $Z$ satisfies an additional condition, then the lower bound on $\mathrm{rank}(\Sigma_0)$ does not depend on the {number} of continuous components of $Z$. To show this result, we modify the arguments of \textcite{horowitz1996direct} to the current framework. 

\begin{theorem}\label{theorem:rank_Sigma3}
Suppose the conditions of Theorem \ref{theorem:rank_Sigma2} and that $Var(Z_B)$ is full rank. If, in addition, the conditional support of $[\gamma_0,\delta_0]^\intercal Z$ given $Z_{B}=z_B$ is the same as the
support of $[\gamma_0,\delta_0]^\intercal Z$, then 
$$
\mathrm{rank}({\Sigma}_0)\geq \dim(Z)-\mathrm{rank}(Var([\gamma_0,\delta_0]^\intercal[Z_{A}^\intercal,0^\intercal]^\intercal)).
$$
Furthermore if $\delta_0^\intercal Z$ is discrete, then
$$
\mathrm{rank}({\Sigma}_0)\geq \dim(Z)-J_1.
$$
\end{theorem}

Relative to Theorem \ref{theorem:rank_Sigma2}, Theorem \ref{theorem:rank_Sigma3} provides an improved lower bound by depending on the length of $Z$ instead of the number of continuous components in $Z$. This improved bound is available when $(\gamma_0,\delta_0)^\intercal Z$ satisfies a rectangular support assumption.

\section{Estimation}\label{sec:estimation}

In this section, we introduce our estimator for $\theta_0$. Our method is motivated by the computational difficulty of an available estimator $\hat\theta^*=\arg\max_{\theta\in\Theta}\hat{Q}(\theta)$, where $\Theta$ is a subset of a Euclidean space and $\hat{Q}:\Theta\rightarrow\mathbb{R}$ is a sample criterion function. As discussed previously, in many cases the estimator $\hat{\theta}^*$ is computationally heavy, or may even be computationally infeasible in practice. For example, maximum likelihood estimation of finite-mixture dynamic discrete models is considered extremely computationally costly, to such a degree that alternative estimators are often preferred \parencite[e.g.,][]{arcidiacono2011conditional}.

Our estimator $\hat{\theta}$ for $\theta_0$ is constructed in the following two steps: 
\begin{itemize}
\item[]\textbf{Step 1:} Estimate $\Sigma_0$ with $\hat{\Sigma}$, and compute
$$
\tilde\theta=\argmax_{\theta\in\Theta:\  \hat{\Sigma}\boldsymbol{\gamma}(\theta)=0}\hat{Q}\left(\theta\right).
$$
\item[]\textbf{Step 2:} Estimate $\theta_0$ with $\hat{\theta}$, computed as follows. Given $L\in\mathbb{N}$ and  $\tilde\theta$ from Step 1, 
\begin{eqnarray*}
\tilde\theta_1&=&\tilde\theta-\hat{Q}^{(2)}(\tilde\theta)^{-1}\hat{Q}^{(1)}(\tilde\theta)\\
\tilde\theta_2&=&\tilde\theta_1-\hat{Q}^{(2)}(\tilde\theta_1)^{-1}\hat{Q}^{(1)}(\tilde\theta_1)\\
&\vdots&\\
\hat\theta&=&\tilde\theta_{L-1}-\hat{Q}^{(2)}(\tilde\theta_{L-1})^{-1}\hat{Q}^{(1)}(\tilde\theta_{L-1}),
\end{eqnarray*}
where $\hat{Q}^{(1)}(\theta)$ and $\hat{Q}^{(2)}(\theta)$ are the first and second derivatives of $\hat{Q}(\theta)$. 
\end{itemize}

In the first step, we form a preliminary estimator $\tilde{\theta}$ by maximizing the sample criterion function subject to the estimated constraints $\hat{\Sigma}\boldsymbol{\gamma}(\theta)=0$. Since these are (low dimensional) linear equality constraints, they can be easily imposed in practice with negligible computational cost (see Section \ref{app:mcmcsigma} for one approach). The second step consists of $L$ Newton-Raphson iterates from the preliminary estimator $\tilde{\theta}$. The main result of this section (Theorem \ref{theorem:main}) states that the number of Newton-Raphson iterates controls the rate at which $\hat{\theta}-\hat{\theta}^*$ converges to zero as  sample size $n$ diverges (i.e., $n\rightarrow\infty$). The remainder of this section is dedicated to showing this result, which will use two additional assumptions.

The first step of our estimator solves a maximization problem subject to the estimated constraint $\hat{\Sigma}\gamma=0$. Naturally, we require that the estimated constraint $\hat{\Sigma}\gamma=0$ provides a good approximation to ${\Sigma}_0\gamma=0$, which we formalize in Assumption \ref{assn:tildeSigma_converge}.

\begin{assumption}\label{assn:tildeSigma_converge}
$\hat{\Sigma}-{\Sigma_0}=o_p(1)$ and $\Pr(\mathrm{rank}(\hat{\Sigma})=\mathrm{rank}(\Sigma_{0}))=1+o(1)$.
\end{assumption}

The first part of Assumption \ref{assn:tildeSigma_converge} states that $\hat{\Sigma}$ is consistent for ${\Sigma_0}$. Notably, the rate of convergence need not be known by the econometrician. In particular, we allow the rate of convergence to be arbitrarily slow: Theorem \ref{theorem:main} implies that even if the convergence rate is slow, only moderate increases in $L$ are required to attain fast convergence between our estimator and the computationally intensive estimator. Many nonparametric methods can achieve consistent estimation (e.g., kernel smoothing, nearest neighbor, splines, or series estimators). In Section \ref{sec:const_tildeSigma}, we provide conditions for consistent estimation using kernel smoothing \parencite{ahn2018simple}.


The second part of Assumption \ref{assn:tildeSigma_converge} states that $\hat{\Sigma}$ is rank-consistent, which ensures that $\hat{\Sigma}\gamma=0$ imposes the same number of linearly independent constraints as $\Sigma_0\gamma=0$ with probability approaching one. Given a consistent estimator $\tilde\Sigma$, which may or may not have the same rank as $\Sigma_0$, one may construct a rank-consistent estimator by a low rank approximation.\footnote{Instead of this approach of $\hat\Sigma$, we may be able to apply a rank estimator, e.g., in \textcite{chen2019improved}. Since our results rely only on the convergence rate of $\tilde\theta$ and the rank is correctly estimated with probability approaching one, we conjecture that estimating the rank does not change our main result.} To explain, let $\hat\lambda_1\geq\cdots\geq\hat\lambda_K$ be the eigenvalues of $\tilde{\Sigma}$, and $\hat\nu_1,\cdots,\hat\nu_K$ be the corresponding eigenvectors. Define the low-rank approximation
\begin{equation}\label{eq:hatSigma_define}
\hat{\Sigma}\equiv[\hat\nu_1,\cdots,\hat\nu_K]^\intercal\mathrm{diag}\left(\hat\lambda_1\cdot 1\{\hat\lambda_1>\kappa\},\ldots,\hat\lambda_K\cdot 1\{\hat\lambda_K>\kappa\}\right)[\hat\nu_1,\cdots,\hat\nu_K],
\end{equation}
where $\kappa$ is a threshold value. The following result (Lemma \ref{lemma:Sigma_estimation}) states that the low-rank approximation $\hat{\Sigma}$ satisfies Assumption \ref{assn:tildeSigma_converge} as long as $\kappa$ converges to zero slowly. 

\begin{lemma}\label{lemma:Sigma_estimation}
If $\tilde{\Sigma}-\Sigma_{0}=o_{p}(\kappa)$ for $\kappa=o(1)$, then $\hat{\Sigma}$ defined in equation \eqref{eq:hatSigma_define} satisfies Assumption \ref{assn:tildeSigma_converge}.
\end{lemma}

It is worthwhile noting that any estimator for $\Sigma_0$ satisfying Assumption \ref{assn:tildeSigma_converge} will be constructed using estimators of $\Pi_0$ and $\delta_0$. Typically, both the reduced form parameter $\Pi_0$ and the nuisance parameter $\delta_0$ are identified directly from the data, and can thus be estimated without constructing the function $\hat{Q}$. For example, in the Monte Carlo exercise of Section \ref{sec:appl}, $\delta_0$ is known \textit{a priori}, and $\Pi_0$ are conditional probabilities which we estimate via kernel regression.

The computationally intensive estimator $\hat\theta^*$ is an example of an extremum estimator. Assumption \ref{assn:unif_converge} imposes mild regularity conditions that are typical in extremum estimation problems.

\begin{assumption}\label{assn:unif_converge}
(i) $\Theta$ is compact, and $\theta_0$ is an interior point of $\Theta$. 
(ii) $\theta_0$ is the unique maximizer of $Q_0(\theta)$ over $\theta\in\Theta$. 
(iii) $Q_0(\theta)$ is twice continuously differentiable such that  the first derivative $Q_0^{(1)}(\theta)$ is bounded and that the second derivative $Q_0^{(2)}(\theta)$ is non-singular at $\theta=\theta_0$. 
(iv) $\sup_{\theta\in\Theta}\|\hat{Q}(\theta)-Q_0(\theta)\|=o_p(1)$ and 
$\|\hat{Q}^{(1)}(\theta_0)-Q_0^{(1)}(\theta_0)\|=O_p(n^{-1/2})$.
(v) There is a neighborhood $\mathcal{N}$ of $\theta_0$ such that $\hat{Q}\left(\theta\right)$ is twice differentiable in $\mathcal{N}$ with 
$\sup_{\theta\in\mathcal{N}}\|\hat{Q}^{(2)}(\theta)-Q_0^{(2)}(\theta)\|=o_p(1)$.
\end{assumption}

We now state the main theoretical result of this paper. 
 
\begin{theorem}\label{theorem:main}
Under Assumptions \ref{assn:key_equivalence}-\ref{assn:unif_converge},
$$
\hat\theta-\hat{\theta}^*=O_p(\max\{\|\hat{\Sigma}-{\Sigma_0}\|,n^{-1/2}\}^{2^L}).
$$
\end{theorem}

Theorem \ref{theorem:main} states that our estimator $\hat{\theta}$ is asymptotically equivalent to the computationally more intensive estimator $\hat{\theta}^*$. In particular, that the difference $\hat\theta -\hat\theta^*$ converges to zero at the rate $\max\{\|\hat{\Sigma}-{\Sigma_0}\|,n^{-1/2}\}^{2^L}$.  Let us now provide some intuition for the rate of convergence. First, the term $\max\{\|\hat{\Sigma}-{\Sigma_0}\|,n^{-1/2}\}$ represents the convergence rate of $\tilde\theta-\hat\theta^*$, i.e., the difference between the start-up and target estimators for the Newton-Raphson iterations. The convergence rate can be understood as follows. Because $\hat\Sigma\tilde\theta=0$, the difference $\tilde\theta-\hat\theta^*$ is proportional to $\hat{\Sigma}\hat{\theta}^*$ whose convergence rate depends on $\hat{\Sigma}-{\Sigma}_0$ and $n^{-1/2}$ (from $\hat{\theta}^*-\theta_0$). Second, the exponent $2^L$ represents the effect of $L$ Newton-Raphson iterations from the start-up estimator $\tilde{\theta}$. As in \textcite{robinson1988stochastic}, the rate of convergence of $\hat\theta$ to $\hat{\theta}^*$ increases exponentially in the number of Newton-Raphson updates $L$.

A practical consideration for our estimator is how to choose the number of Newton-Raphson iterations $L$. Our main theoretical result (Theorem \ref{theorem:main}) suggests that $L$ should be chosen to achieve the desired rate of convergence between $\hat\theta$ and $\hat\theta^*$. For example, if $\hat\theta^*$ is justified by first-order asymptotics, then $L$ can be chosen to achieve first-order asymptotic equivalence between $\hat\theta$ and $\hat\theta^*$, which is attained with one Newton-Raphson update when $\hat{\Sigma}-{\Sigma_0}=o_p(n^{-1/4})$. If $\hat\theta^*$ has desirable higher-order asymptotic properties, then $L$ can be set to a larger number. Importantly, because $L$ impacts the rate of convergence through the exponent $2^L$, fast convergence of $\hat\theta-\hat\theta^*$ can be attained for moderate $L$. Of course, extra Newton-Raphson iterations impose additional computation costs. However, our experience in simulations suggests that the computational cost of Newton-Raphson updates (i.e., Step 2 of our estimator) may be small relative to solving the constrained optimization problem (i.e., Step 1). Overall, consideration of theoretical and empirical aspects suggests choosing $L$ as small as possible to achieve the desired degree of asymptotic equivalence.

\section{Monte Carlo simulations}\label{sec:appl}

This section investigates the performance of our proposed estimator in a Monte Carlo simulation. The two primary objectives in this section are, first, to quantify the computational benefits of our approach and, second, to provide empirical support for our main asymptotic equivalence result (Theorem \ref{theorem:main}). Section \ref{app:mcmc} contains further details on the Monte Carlo design and implementation.

Our Monte Carlo design is based on the empirical setting of \textcite{toivanen2005market}, which analyzes firm entry into the U.K. fast food market between 1991 and 1995. Restricting attention to the largest two firms, their analysis divides the U.K. into 422 local markets and records information about each market and the firms' decisions on how many stores to operate in each market. To fit it within the scope of this paper, we model a single firm's decision as a dynamic discrete choice problem with the profit function in the spirit of \textcite{bresnahan1991entry,toivanen2005market,aguirregabiria2020identification}.

\subsection{A dynamic model of firm entry}\label{sec:appl-model}

In each period and geographic market, a firm decides whether to open an additional store, upon observation of the state variables. The firm's decision in market $i$ and time $t$ is $A_{i,t}\in\{0,1\}$, which takes value $1$ if the firm opens a store in market $i$ at time $t$, and $0$ otherwise. In each period $t$, the vector of state variables known by the firm in market $i$ is $(N_{i,t},W_i^\intercal,\lambda_{i},\epsilon_{i,t})^\intercal$ where $N_{i,t}$ is the number of incumbent stores (that is, prior to the realization of $A_{i,t}$) that we assume is bounded above by 3, $W_i$ are variables that affect the size of market $i$, $\lambda_{i}$ is the market type, and $\epsilon_{i,t}$ is an idiosyncratic shock. Firms are assumed to be forward looking, choosing $A_{i,t}$ to maximize expected discounted profits. We assume that $(N_{i,t},W_i^\intercal)^\intercal$ and $(\lambda_i,\epsilon_{i,t})^\intercal$ are observed and unobserved to the econometrician, respectively.

Denoting the observed state variable $(N_{i,t},W_i^\intercal)^\intercal$, period profits from opening an additional store in market $i$ at time $t$ (i.e., $A_{i,t}=1$) are equal to
$$
\left(\lambda_i + \theta_W^\intercal W_i\right) - \left(\theta_{FC} N_{i,t} + \theta_{EC} \mathsf{1}(N_{i,t}=0) + \epsilon_{i,t}\right).
$$
The first component $(\lambda_i + \theta_W^\intercal W_i)$ of the period profits is the marginal revenue from opening an additional store, i.e., the market size. Following the literature \parencite{toivanen2005market,aguirregabiria2020identification}, we allow market size to depend on an unobserved market type ($\lambda_i$) and a long vector of continuous demographic and socioeconomic variables. For example, in \textcite{toivanen2005market}, market size depends on total population, youth population, and pensioner-age population. In \textcite{aguirregabiria2020identification} market size depends additionally on population density, the local unemployment rate and local GDP per capita. Inspired by these papers, we set $W_i\in\mathbb{R}^{\dim(W_i)}$ with $\dim(W_i)=9$. Also, we impose that $W_i$ and $\lambda_i$ are statistically independent. The second component $\left(\theta_{FC} N_{i,t} + \theta_{EC} \mathsf{1}(N_{i,t}=0) + \epsilon_{i,t}\right)$ represents the marginal cost of opening an additional store, which depends on the firm's local experience. Also following \textcite{toivanen2005market,aguirregabiria2020identification}, we assume $\epsilon_{i,t}$ is an opening cost shock that is known to follow the standard normal distribution. The period payoff from \textit{not} opening an additional store (i.e., $A_{i,t}=0$) is normalized to zero. Firms discount future payoffs at a discount factor $\beta=0.95$, which we assume to be known. Finally, observe that $N_{i,t+1}$ is uniquely determined by $(N_{i,t},A_{i,t})$, i.e., $N_{i,t+1}=\max\{N_{i,t}+A_{i,t},3\}$, so the transition of the state variable is deterministic. We collect the part of the structural parameter that governs the flow payoff as $\theta=(\theta_W^\intercal,\theta_{FC},\theta_{EC})^\intercal\in\mathbb{R}^{9+1+1}$. The remaining component of the structural parameter is the distribution of types $F_{\lambda}$. 

We now discuss Assumption \ref{assn:key_equivalence} and the dimension reduction of our proposed method. Let $Z_{i,t}=(W_i^\intercal,t)^\intercal$, $\gamma(\theta)=(\theta_W^\intercal, 0)^\intercal$, $\delta_0=(0_{\dim(W)}^\intercal,1)^\intercal$, and $\Pi_0(z)=Pr(N_{i,t+1}\ge 1\mid W_i=w)$. With these definitions, Assumption \ref{assn:key_equivalence} holds if, for every pair $z$ and $\tilde{z}$ with $t=\tilde{t}$,
\begin{equation}\label{eq:assumption1_section5}
\Pi_0(z)>\Pi_0(\tilde{z}) \iff \theta_W^\intercal(w-\tilde{w})>0. 
\end{equation}
That is, the probability of having a store operating in a given period is higher if and only if the payoff index $\theta_W^\intercal W_i$ is higher. In Appendix \ref{appendix:monotoinicityinSection5} we show that, under some conditions, equation \eqref{eq:assumption1_section5} holds. Thus $\Sigma_0\boldsymbol{\gamma}(\theta_0)=0$ where $\Sigma_0=E[(Z-\tilde{Z})(Z-\tilde{Z})^\intercal\mid \Pi_0(Z)=\Pi_0(\tilde{Z}),\delta_0^\intercal Z=\delta_0^\intercal\tilde{Z}]$  for $Z$, $\tilde{Z}$ independent draws of $Z_{i,t}$. Given that $\mathrm{Supp}(W_i)$ has a non-empty interior and $\delta_{0}^\intercal Z_{i,t}=t$ is discrete, the second part of Theorem \ref{theorem:rank_Sigma2} applies and $\mathrm{rank}(\Sigma_0)\ge 9-1=8$, so that the identity $\Sigma_0\boldsymbol{\gamma}(\theta_0)=0$ may reduce the dimension of the optimization problem by $\dim(W_i)-1=8$.

Finally, we describe our choice of simulation parameters. In our design, we observe the set $\{N_{i,t},W_i,A_{i,t}:t=1,2\ldots,8\}_{i=1}^n$ for $n=100$, $200$, $350$, and $500$  i.i.d. markets (recall in the dataset of \textcite{toivanen2005market}, $n=422$ and $T=5$). Our simulation results are based on 100 draws from this data generating process with the following value of the structural parameter: $\mathrm{Supp}(\lambda_i)=\{0.1, 1\}$, $\Pr(\lambda_i=1)=0.63$, and $\theta_0=((-0.3,-0.2,-0.1,0.1,,0.2,0.3,0.4,0.5,-0.6),0.5,0.5)^\intercal$. We initialize $N_{i,1}=0$ and draw $W_i$ from the uniform distribution over $[0,1]^{\dim(W_i)}$.

\subsection{Estimation}\label{ssec:appl_est}

We consider two target estimators for our method: one based upon \textcite{arcidiacono2011conditional}, the other upon \textcite{bajari2011simple}. Both of our target estimators are extremum estimators that use the sample log-likelihood function as the criterion, that is
\begin{equation}\label{eq:hatq}
\sum_{i=1}^{n}\log \int \prod_{t=1}^8 \ell_{i,t}(v,\theta) dF_{\lambda}(v),
\end{equation}
where the likelihood contribution $\ell_{i,t}(v,\theta)$ is the model-implied probability of $A_{i,t}$ conditional upon $N_{i,t},W_i$ and $\lambda_i=v$ evaluated at $\theta$. The estimators differ in their approach to modeling the distribution of $\lambda_i$. Specifically, while both approaches model ${\lambda_i}$ as a discrete random variable, \textcite{arcidiacono2011conditional}'s estimator fixes the number of support points but not their locations, whereas \textcite{bajari2011simple}'s estimator specifies a (possibly large) number of fixed points that are known to contain the support.\footnote{The estimator of \textcite{bajari2011simple} approach has also been analyzed as a nonparametric estimator in certain discrete choice problems \parencites[e.g.,][]{fox2016simple,bunting2020}.} Section \ref{app:mcmc-est} provides details on our definition of these estimators. For the estimator of \textcite{arcidiacono2011conditional}, we correctly specify that $\lambda_i$ has two unknown points of support. For the estimator based on \textcite{bajari2011simple}, we allow the support to contain any or all of 21 uniformly spaced points from $-0.5$ to $1.5$.  Henceforth, we refer to our formulation of \textcite{arcidiacono2011conditional}'s estimator as $\hat\theta^*_{EM}$, and the estimator based upon \textcite{bajari2011simple} as $\hat\theta^*_{H}$.

For each target estimator, our estimator is constructed in the following two steps:
\begin{enumerate}
\item Form $\tilde{\theta}$ by solving the target estimator's optimization problem subject to the constraint $\hat{\Sigma}{\theta}=0$, where $\hat{\Sigma}$ is constructed according to Sections \ref{sec:const_tildeSigma} and \ref{app:mcmcsigma}.
\item Form $\hat{\theta}$ by taking $L=50$ Newton-Raphson updates from $\tilde{\theta}$ towards the target estimator of $\theta_0$.
\end{enumerate}
For the target estimators $\hat\theta^*_{H}$ and $\hat\theta^*_{EM}$, we denote our corresponding estimator as $\hat\theta_{H}$ and $\hat\theta_{EM}$, and the first step estimators as $\tilde\theta_{H}$ and $\tilde\theta_{EM}$, respectively. To illustrate the computational comparison with two standard approaches to dynamic discrete choice estimation, we also compute $\hat\theta^*_{H}$ and $\hat\theta^*_{EM}$ for each simulated dataset.

Estimation of this model may be computationally intensive for two main reasons. First, the choice probabilities $\ell_{i,t}(v,\theta)$ are generated as the solution to a dynamic programming problem, which is solved by iterating a contraction mapping until convergence separately for each candidtate $(v,\theta)$.\footnote{To solve the model, we iterate on the conditional choice probability mapping \parencite[equation 2.3]{bugni2021iterated}.} Second, the presence of unobserved types $\lambda_i$  means that the observed data is an unknown mixture of type-specific choice models. Even if $\lambda_i$ is assumed to have finite support, its presence means that the log-likelihood function may not be globally concave, which creates the risk that the optimization algorithm converges to a non-global optimum. To mitigate this risk, one may rerun the optimization algorithm a number of times, each run starting from a different initial value \parencite[p. 182]{robert1999monte}. One ideal approach would be to initialize the estimator at each element of $\times_{d=1}^D \{i_{d,1},i_{d,2},\ldots,i_{d,P}\}$ where $D$ is the dimension of the optimization problem, and $i_{d,p},p=1,\dots,P$ are the vertices of a one dimensional grid. We choose $P=11$ in this section. For moderate $D$, this ideal approach quickly becomes infeasible. To see this, the ideal approach described above would require ${P}^3$ initializations to form $\tilde\theta$ with $D=3$, and $P^{11}$ initializations to form $\hat\theta^*$ with $D=11$. Instead, we proceed by randomly sampling $2\times D$ times from $\times_{d=1}^D \{i_{d,1},\ldots,i_{d,P}\}$ without replacement, where we choose $i_{d,p},p=1,\dots,P$ to be evenly spaced points around an initial guess based on the (parametric) pseudo maximum likelihood estimator. See Section \ref{app:startval} for a formal specification of the initial points. We then define $\tilde\theta$ and $\hat\theta^*$ as the maximizer of the sample criterion over the set of initial points.\footnote{More precisely, for $j=H,EM$,  if  $\mathcal{I}$ is the set of initial points and $\tilde\theta_{j}^{(i)}$ (resp. $\hat\theta_{j}^{*,(i)}$) is the solution to the constrained (resp. unconstrained) optimization problem initialized at $i\in\mathcal{I}$, then $\tilde\theta_{j}=\argmax_{\{\tilde\theta_{j}^{(i)}\colon i\in\mathcal{I} \}}\hat{Q}(\theta)$ (resp. $\hat\theta_{j}^*=\argmax_{\{\hat\theta_{j}^{*,(i)}\colon i\in\mathcal{I} \}}\hat{Q}(\theta)$).}

\subsection{Simulation results}

Table \ref{tab:runtime1} displays the mean computation time in minutes for our first step estimator $\tilde\theta$, our final estimator $\hat\theta$, and the target estimator $\hat\theta^*$, for each sample size and estimation method.\footnote{Simulations were run in the Julia programming language, and replication files are available at \href{https://doi.org/10.5281/zenodo.15116658}{https://doi.org/10.5281/zenodo.15116658}. Results for $\hat\theta_H$ were generated using Julia 1.11 on a standard desktop computer, and results for $\hat\theta_{AM}$ were generated using Julia 1.10 on the JuliaHub cloud-based server. The replication files contain a complete description of the two environments.} A number of observations can be made. First, our estimator is substantially faster than the standard estimator.  For example, for $n=500$, our method for targeting the estimator based on \textcite{bajari2011simple} is over 20 times faster than the standard approach for $n=500$, which translates to reducing the computation time from around 9 hours to just over 25 minutes. For $n=500$, our method for targeting \textcite{arcidiacono2011conditional}'s estimator is over 16 times faster than the standard approach, which translates to reducing the computation time from over 9 hours to around one half hour. Second, the computational cost of the Newton-Raphson iterates (i.e., Step 2 of our estimator) is minor relative to solving the constrained optimization problem (i.e., Step 1 of our estimator). Namely, the difference between the computational times for obtaining $\tilde\theta$ and for obtaining $\hat\theta$ is small. Third, we observe that the computational burden increases approximately linearly with sample size. This means the computational time savings are substantially bigger when $n$ is large---that is, the gain from our method is larger for the  `harder' computational problem.

\begin{table}[ht]
\centering
\begin{tabular}{r|rrr|rrr}
  \hline
   \rule{0pt}{2.5ex}    $n$ & $\tilde\theta_{H}$& {$\hat\theta_{H}$} & {$\hat\theta_{H}^*$} & $\tilde\theta_{EM}$ & $\hat\theta_{EM}$ & $\hat\theta_{EM}^*$ \bigstrut[t]\\
  \hline
100 &   3.7 &   5.3 & 115.0 &   6.3 &   6.9 & 130.6 \\ 
  200 &   7.0 &  10.1 & 222.7 &  13.1 &  13.4 & 260.1 \\ 
  350 &  11.8 &  17.2 & 377.3 &  23.0 &  23.3 & 417.4 \\ 
  500 &  17.9 &  25.6 & 536.3 &  33.4 &  33.7 & 556.0 \\ 
   \hline
\end{tabular}
\caption{Mean (over 100 replications) run time in minutes, 
        by sample size and estimator. For $j = H,EM$, the first step of our estimator,
         our estimator and the target estimator, are denoted by $\tilde\theta_{j}$,
          $\hat\theta_{j}$, and $\hat\theta_{j}^*$, respectively. $H$ denotes the 
          estimator based upon \textcite{bajari2011simple} and $EM$ denotes the estimator 
          based upon \textcite{arcidiacono2011conditional}.} 
\label{tab:runtime1}
\end{table}

The computational gains from our method described in Table \ref{tab:runtime1} can be decomposed as the product of two effects. Namely, our method reduces the number of grid points (and thus initializations) designed to cover the parameter space, and also reduces the computational burden \textit{per initialization} of the optimization algorithm.\footnote{As described above, we choose 7 ($=2\times3+1$) and 23 ($=2\times 11+1$) random starting values for the constrained $\hat\theta$ and unconstrained $\hat\theta^*$, respectively. However, clearly, the computational gain from reducing the dimension of the grid can be made exponentially larger by requiring more points for each dimension of the grid. We also point out that parallelization can help mitigate this part of the computational burden.} We explore this decomposition in Table \ref{tab:runtime2} which presents the average runtime per initialization, for each sample size and estimation method. We observe that, per initialization, our method is around 6 times faster than the standard, unconstrained approach on average.  For example, for $n=500$, our method of targeting the estimator based upon \textcite{bajari2011simple} is around 6 times faster per iteration than the standard approach, which translates to reducing the computation time from 24 minutes to around 4 minutes per run. Similarly, for $n=500$, the constrained EM algorithm is around 4.5 times faster per iteration than the unconstrained approach, which translates to reducing the computation time from around 25 minutes to 6 minutes per run.

\begin{table}[h]
\centering
\begin{tabular}{r|rr|rr}
  \hline
   \rule{0pt}{2.5ex}    $n$ & $\tilde\theta_{H}^{(i)}$ & {$\hat\theta_{H}^{*,(i)}$} & $\tilde\theta_{EM}^{(i)}$ & $\hat\theta_{EM}^{*,(i)}$ \bigstrut[t]\\
  \hline
100 &  0.8 &  5.2 &  1.1 &  5.9 \\ 
  200 &  1.4 & 10.1 &  2.2 & 11.7 \\ 
  350 &  2.4 & 17.1 &  3.9 & 18.8 \\ 
  500 &  3.6 & 24.3 &  5.7 & 25.1 \\ 
   \hline
\end{tabular}
\caption{Mean (over 100 replications) run time in minutes of the run time per initialization of the optimization
         algorithm, by sample size and estimator. For $j = EM, H$, $\tilde\theta_{j}^{(i)}$, and $\hat\theta_{j}^{*,(i)}$
          refer         to one initialization of the optimization algorithm to compute $\tilde\theta_j$ and 
          $\hat\theta_{j}^{*}$, respectively. $H$ denotes the estimator based upon \textcite{bajari2011simple} 
          and $EM$ denotes the estimator based upon \textcite{arcidiacono2011conditional}.} 
\label{tab:runtime2}
\end{table}

We now turn to results on the convergence between our estimator and the corresponding target estimators. For each target estimator, Table \ref{tab:conv1} shows that the $\sqrt{n}$-scaled root mean squared error of $\hat\theta-\hat\theta^*$ tends to decrease with sample size. Thus it appears that the difference of the estimators is converging to zero at faster than $\sqrt{n}$ rate, which provides empirical support for our main theorem (Theorem \ref{theorem:main}).\footnote{For sake of comparison, we report the $\sqrt{n}$-scaled root mean squared error of $\tilde\theta-\hat\theta^*$ in Table \ref{tab:conv2} of Appendix \ref{sec:additional_MC}. In general, the $\sqrt{n}$-scaled root mean squared error of $\hat\theta-\hat\theta^*$ is much smaller than that of $\tilde\theta-\hat\theta^*$.}

\begin{table}[h]
\centering
\begin{tabular}{l|rrrr|rrrr}
  \hline
   \rule{0pt}{2.5ex}  & \multicolumn{4}{c|}{$\sqrt{n}(\hat\theta_{H}-\hat\theta^*_{H})$} & \multicolumn{4}{c}{$\sqrt{n}(\hat\theta_{EM}-\hat\theta^*_{EM})$} \bigstrut[t] \\
 \multicolumn{1}{r|}{$n:$} & 100 & 200 & 350 & 500 & 100 & 200 & 350 & 500 \\ 
   \hline
$\theta_{W,1}$ & 0.290 & 0.200 & 0.205 & 0.203 & 0.148 & 0.169 & 0.013 & 0.010 \\ 
  $\theta_{W,2}$ & 0.272 & 0.191 & 0.211 & 0.189 & 0.081 & 0.065 & 0.014 & 0.011 \\ 
  $\theta_{W,3}$ & 0.275 & 0.186 & 0.209 & 0.167 & 0.310 & 0.029 & 0.011 & 0.016 \\ 
  $\theta_{W,4}$ & 0.218 & 0.199 & 0.184 & 0.165 & 0.070 & 0.089 & 0.014 & 0.012 \\ 
  $\theta_{W,5}$ & 0.371 & 0.164 & 0.183 & 0.159 & 0.189 & 0.126 & 0.013 & 0.012 \\ 
  $\theta_{W,6}$ & 0.278 & 0.162 & 0.166 & 0.160 & 0.119 & 0.084 & 0.014 & 0.012 \\ 
  $\theta_{W,7}$ & 0.241 & 0.134 & 0.155 & 0.139 & 0.093 & 0.072 & 0.014 & 0.010 \\ 
  $\theta_{W,8}$ & 0.418 & 0.152 & 0.137 & 0.142 & 0.181 & 0.158 & 0.014 & 0.011 \\ 
  $\theta_{W,9}$ & 0.315 & 0.228 & 0.222 & 0.213 & 0.151 & 0.106 & 0.013 & 0.013 \\ 
  $\theta_{EC}$ & 0.230 & 0.177 & 0.182 & 0.178 & 0.115 & 0.143 & 0.010 & 0.012 \\ 
  $\theta_{FC}$ & 0.513 & 0.402 & 0.437 & 0.412 & 0.208 & 0.290 & 0.031 & 0.039 \\ 
   \hline
\end{tabular}
\caption{Empirical mean squared error of $\sqrt{n}(\hat\theta_{H}-\hat\theta^*_{H})$ and $\sqrt{n}(\hat\theta_{EM}-\hat\theta^*_{EM})$, for each sample size and element of $\theta$, calculated over the 100 replications.  $H$ denotes the estimator based upon \textcite{bajari2011simple} and $EM$ denotes the estimator based upon \textcite{arcidiacono2011conditional}.} 
\label{tab:conv1}
\end{table}

To conclude, as established in Tables \ref{tab:runtime1} and \ref{tab:runtime2}, our estimator can be computed substantially faster than the target estimator. Table \ref{tab:conv1} provides evidence that the two estimators are asymptotically equivalent, consistent with our main result (Theorem \ref{theorem:main}). Taken together, the two findings demonstrate that our method can be used to complement both the estimator of \textcite{arcidiacono2011conditional} and \textcite{bajari2011simple}: regardless of which approach to modeling $F_\lambda$ is preferred, our method can be used to form an asymptotically equivalent estimator, with a much lighter computational burden.

\section{Conclusion}\label{sec:concl}

In this paper we provide a method to simplify estimation of structural economic models whose reduced form parameters are an invertible function of a number of linear indices. A leading example of such `index invertibility' are dynamic discrete choice models with persistent unobserved heterogeneity and partially linear flow payoffs. Index invertibility implies a set of equality constraints which restrict the structural parameter of interest to belong in a subspace of the parameter space. We propose an estimator that imposes the equality constraints, and show it is asymptotically equivalent to the unconstrained estimator. The proposed constrained estimator may be computationally advantageous due to the effective reduction in the dimension of the optimization problem. Furthermore, we provide a number of results on the extent of effective dimension reduction, and demonstrate our method in Monte Carlo simulations.

\paragraph{Acknowledgments}

We thank Aureo de Paula and the review panel for excellent comments and suggestions that have improved the manuscript. We also thank Yanqin Fan and seminar participants at Boston College, Georgetown, Duke, Simon Fraser and the University of Melbourne for helpful comments. The usual disclaimer applies.

\addcontentsline{toc}{section}{\refname}
\printbibliography

\appendix

\counterwithin*{equation}{section}
\renewcommand\theequation{\thesection.\arabic{equation}}

\section{Proofs}

We use $\dagger$ to denote the Moore-Penrose inverse of a matrix and $\otimes$ to be the Kronecker product. 

\subsection{Identification of \texorpdfstring{$\boldsymbol{\Pi_0}$}{Pi\_0} and 
Proof of Theorem \ref{theorem:A1}}\label{sec:A1_proof}

The function $\Pi_0$ of Section \ref{subsec:ddcm} can be identified by the so-called \textit{lower-dimensional submodel} of \textcite{kasahara2009nonparametric} (see equation (11) in their paper). That is, under the below rectangular support assumption, for any $t\ge2$,
\begin{equation}\label{eq:defin_pi_0_general}
\Pi_0(a,z)=E\left[\frac{f_{Z_t}(Z_t)}{f_{Z_t|Z_{t-1},A_{t-1}}(Z_t|Z_{t-1},A_{t-1})}1\{A_t=a\}\mid Z_t=z\right].
\end{equation}
Since $\Pi_0$ can be expressed as an expectation of a function of the observed data, it may be estimated by any of the many closed-form, consistent estimators of $\Pi_0(a,z)$ which are computationally simple.

\begin{lemma}[\textcite{kasahara2009nonparametric}]
If $f_{Z_t|Z_{t-1},A_{t-1}}(z|Z_{t-1},A_{t-1})>0$ a.s. for every $z$ in the support of $Z_t$, then Eq. \eqref{eq:defin_pi_0_general} holds for every $t\geq 2$.
\end{lemma}
\begin{proof}
The proof of this lemma follows from Section 2.1 of \textcite{kasahara2009nonparametric}, but we provide it for completeness. Let $\mu$ be the dominating measure for the distribution of $(Z_{t-1},A_{t-1})$.

First, we show that the right-hand side of Eq. \eqref{eq:defin_pi_0_general} is the lower-dimensional submodel of \textcite{kasahara2009nonparametric}, that is, 
$$
E\left[\frac{f_{Z_t}(Z_t)}{f_{Z_t|Z_{t-1},A_{t-1}}(Z_t|Z_{t-1},A_{t-1})}1\{A_t=a\}\mid Z_t=z\right]
=\int_{\mathcal{S}(z)} 
\frac{f_{A_{t},Z_{t},A_{t-1},Z_{t-1}}(a,z,a_{t-1},z_{t-1})}{f_{Z_{t}|Z_{t-1},A_{t-1}}(z|z_{t-1},a_{t-1})}
d\mu(z_{t-1},a_{t-1}),
$$
where the above integral is taken over $\mathcal{S}(z)=\{(z_{t-1},a_{t-1}): f_{Z_{t}|Z_{t-1},A_{t-1}}(z|z_{t-1},a_{t-1})>0\}$. By the law of iterated expectations and the definition of conditional densities, we have 
\begin{align*}
&
E\left[\frac{f_{Z_t}(Z_t)}{f_{Z_t|Z_{t-1},A_{t-1}}(Z_t|Z_{t-1},A_{t-1})}1\{A_t=a\}\mid Z_t=z\right]
\\
&=E\left[\frac{\Pr(A_t=a\mid Z_t=z,Z_{t-1},A_{t-1})f_{Z_t}(z)}{f_{Z_{t}|Z_{t-1},A_{t-1}}(Z_{t}|Z_{t-1},A_{t-1})}\mid Z_t=z\right]\\
&=\int_{\mathcal{S}(z)}  
\frac{\Pr(A_t=a\mid Z_t=z,Z_{t-1}=z_{t-1},A_{t-1}=a_{t-1})f_{Z_t}(z)}{f_{Z_{t}|Z_{t-1},A_{t-1}}(z|z_{t-1},a_{t-1})}
f_{Z_{t-1},A_{t-1}\mid Z_t}(z_{t-1},a_{t-1}\mid z)
d\mu(z_{t-1},a_{t-1})\\
&=\int_{\mathcal{S}(z)}  
\frac{f_{A_{t},Z_{t},A_{t-1},Z_{t-1}}(a,z,a_{t-1},z_{t-1})}{f_{Z_{t}|Z_{t-1},A_{t-1}}(z|z_{t-1},a_{t-1})}
d\mu(z_{t-1},a_{t-1}).
\end{align*} 

Then, we show 
$$
\int_{\mathcal{S}(z)} 
\frac{f_{A_{t},Z_{t},A_{t-1},Z_{t-1}}(a,z,a_{t-1},z_{t-1})}{f_{Z_{t}|Z_{t-1},A_{t-1}}(z|z_{t-1},a_{t-1})}
d\mu(z_{t-1},a_{t-1})=\Pi_0(a,z).
$$
Note that 
$$
f_{A_{t},Z_{t},A_{t-1},Z_{t-1}}(a,z,a_{t-1},z_{t-1})
=
E[f_{A_{t},Z_{t},A_{t-1},Z_{t-1}\mid \lambda_t}(a,z,a_{t-1},z_{t-1}\mid \lambda_t)]
$$
and that 
\begin{align*}
f_{A_{t},Z_{t},A_{t-1},Z_{t-1},\lambda_t}(a,z,a_{t-1},z_{t-1},l_t)
&=
f_{A_{t}\mid Z_{t},\lambda_t,A_{t-1},Z_{t-1}}(a\mid z,l_t,a_{t-1},z_{t-1})\\&\quad\times 
f_{Z_{t}\mid \lambda_t, A_{t-1},Z_{t-1}}(z\mid l_t,a_{t-1},z_{t-1})\\&\quad\times 
f_{A_{t-1},Z_{t-1},\lambda_t}(a_{t-1},z_{t-1},l_t).
\end{align*}
Since $A_{t}$ and $(A_{t-1},Z_{t-1})$ are independent given $(Z_{t},\lambda_t)$ and $Z_{t}$ and $\lambda_t$ are independent given $(A_{t-1},Z_{t-1})$, we have 
$$
f_{A_{t},Z_{t},A_{t-1},Z_{t-1}}(a,z,a_{t-1},z_{t-1})
=
E[f_{A_{t}\mid Z_{t},\lambda_t}(a\mid z,\lambda_t)
f_{A_{t-1},Z_{t-1}\mid \lambda_t}(a_{t-1},z_{t-1}\mid \lambda_t)]f_{Z_{t}\mid A_{t-1},Z_{t-1}}(z\mid a_{t-1},z_{t-1}).
$$
Then we have 
\begin{align*}
&
\int_{\mathcal{S}(z)}\frac{f_{A_{t},Z_{t},A_{t-1},Z_{t-1}}(a,z,a_{t-1},z_{t-1})}{f_{Z_{t}|Z_{t-1},A_{t-1}}(z|z_{t-1},a_{t-1})}d\mu(z_{t-1},a_{t-1})
\\
&=
\int_{\mathcal{S}(z)}E[f_{A_{t}\mid Z_{t},\lambda_t}(a\mid z,\lambda_t)
f_{A_{t-1},Z_{t-1}\mid \lambda_t}(a_{t-1},z_{t-1}\mid \lambda_t)]d\mu(z_{t-1},a_{t-1})
\\
&=
E[f_{A_{t}\mid Z_{t},\lambda_t}(a\mid z,\lambda_t)
\int_{\mathcal{S}(z)}f_{A_{t-1},Z_{t-1}\mid \lambda_t}(a_{t-1},z_{t-1}\mid \lambda_t)d\mu(z_{t-1},a_{t-1})].
\end{align*}
where the second equation exchanges the order of integrals. By the rectangular support assumption, $\int_{\mathcal{S}(z)}f_{A_{t-1},Z_{t-1}\mid \lambda_t}(a_{t-1},z_{t-1}\mid \lambda_t)d\mu(z_{t-1},a_{t-1})=1$. As a result, we have 
$$
\int_{\mathcal{S}(z)}\frac{f_{A_{t},Z_{t},A_{t-1},Z_{t-1}}(a,z,a_{t-1},z_{t-1})}{f_{Z_{t}|Z_{t-1},A_{t-1}}(z|z_{t-1},a_{t-1})}d\mu(z_{t-1},a_{t-1})
=
E[f_{A_{t}\mid Z_{t},\lambda_t}(a\mid z,\lambda_t)
]=\Pi_0(a,z).
$$
\end{proof}

\begin{proof}[Proof of Theorem \ref{theorem:A1}]
The only-if part follows from the fact that the choice probability depends on $Z_t$ only through $[\gamma_0,\delta_0]^\intercal Z_t$. We are going to show the if part of this theorem. Suppose for $(z_1,z_2)$ in the support of $Z_t$, $\delta_{0}^\intercal(z_{1}-z_{2})=0$ and $\gamma_0^\intercal (z_1-z_2)\ne 0$. Define $a^*=\argmin_{a\in\mathcal{A}}\gamma_{0}(a)^\intercal(z_{1}-z_{2})$. We will show that $\Pi_{0}(a^*,z_{1})\neq\Pi_{0}(a^*,z_{2})$.

First, by definition, for all $a\in\mathcal{A}$,
\begin{equation}\label{eq:proof_thoerme2_eq1}
\left(\gamma_0(a)-\gamma_0(a^*)\right)^\intercal(z_{1}-z_{2})\geq 0.
\end{equation}
We now show that for at least one $a\in\mathcal{A}$ the inequality in equation \eqref{eq:proof_thoerme2_eq1} is strict. First note that since $\gamma(0)=0$, $\gamma_{0}(a^*)^\intercal(z_{1}-z_{2})\leq 0$.
If $\gamma_{0}(a^*)^\intercal(z_{1}-z_{2})<0$, then $(\gamma_0(0)-\gamma_{0}(a^*))^\intercal(z_{1}-z_{2})>0$. If $\gamma_{0}(a^*)^\intercal(z_{1}-z_{2})=0$, then since $\gamma_0^\intercal (z_1-z_2)\ne 0$, there is at least one $a$ such that $\gamma_{0}(a)^\intercal(z_{1}-z_{2})>0$, and thus $(\gamma_0(a)-\gamma_{0}(a^*))^\intercal(z_{1}-z_{2})>0$.

Therefore, since for any $a$, 
$$
v({a},z_{1},\lambda)-v(a^*,z_{1},\lambda)
\geq 
v({a},z_{2},\lambda)-v(a^*,z_{2},\lambda)
\iff 
\left(\gamma_0(a)-\gamma_0(a^*)\right)^\intercal z_{1}
\geq
\left(\gamma_0(a)-\gamma_0(a^*)\right)^\intercal z_{2}
$$
we have that
\begin{align*}
&\left\{\epsilon_t\in\mathbb{R}^{J_1+1}\colon~\forall{a}\in\mathcal{A},~\epsilon_{t}(a^*)-\epsilon_{t}({a})\ge{}v({a},z_{1},\lambda_t)-v(a^*,z_{1},\lambda_t)\right\}
\\
&\subsetneq
\left\{\epsilon_t\in\mathbb{R}^{J_1+1}\colon~\forall{a}\in\mathcal{A},~\epsilon_{t}(a^*)-\epsilon_{t}({a})\ge v({a},z_{2},\lambda_t)-v(a^*,z_{2},\lambda_t)\right\}.
\end{align*}
By the independence between $\epsilon_{t}$ and $(Z_t,\lambda_t)$ and the full support of $\epsilon_t$, 
\begin{align*}
&\Pr\left(\left\{\epsilon_t\in\mathbb{R}^{J_1+1}\colon~\forall{a}\in\mathcal{A},~\epsilon_{t}(a^*)-\epsilon_{t}({a})\ge{}v({a},z_{1},\lambda_t)-v(a^*,z_{1},\lambda_t)\right\}\mid Z_t=z_1,\lambda_t=v\right)
\\
&< 
\Pr\left(\left\{\epsilon_t\in\mathbb{R}^{J_1+1}\colon~\forall{a}\in\mathcal{A},~\epsilon_{t}(a^*)-\epsilon_{t}({a})\ge v({a},z_{2},\lambda_t)-v(a^*,z_{2},\lambda_t)\right\}\mid  Z_t=z_2,\lambda_t=v\right),
\end{align*}
which implies $\Pi_{0}(a^*,z_{1})\ne\Pi_{0}(a^*,z_{2})$.
\end{proof}

\subsection{Proof of Theorem \ref{theorem:rank_Sigma2}}
 
\begin{proof} 
We can express $[\gamma_0,\delta_0]^\intercal[Z_{A}^\intercal,0^\intercal]^\intercal$ as 
$$
[\gamma_0,\delta_0]^\intercal[Z_{A}^\intercal,0^\intercal]^\intercal=\mathbf{M}_1\mathbf{M}_2Z_{A}\mbox{ almost surely}
$$
for some $\mathbf{M}_1\in\mathbb{R}^{(J_1+J_2)\times \mathrm{rank}(Var([\gamma_0,\delta_0]^\intercal[Z_{A}^\intercal,0^\intercal]^\intercal))}$ and $\mathbf{M}_2\in\mathbb{R}^{\mathrm{rank}(Var([\gamma_0,\delta_0]^\intercal[Z_{A}^\intercal,0^\intercal]^\intercal))\times \dim(Z_{A})}$. Let $\bar{\nu}_1,\ldots,\bar{\nu}_{R_A}$ be $R_A=\dim(Z_{A})-\mathrm{rank}(Var([\gamma_0,\delta_0]^\intercal[Z_{A}^\intercal,0^\intercal]^\intercal))$ linearly independent vectors in the column space of 
$$
\left(\begin{array}{cc}
I-\mathbf{M}_2^\dagger\mathbf{M}_2\\
O
\end{array}\right),
$$
which exist since the rank of the above matrix is at least $\dim(Z_A)-\mathrm{rank}(\mathbf{M}_2)$. Note that $[\gamma_0,\delta_0]^\intercal\bar{\nu}_r=0$ for every $r=1,\ldots,R_A$. By the expression of ${\Sigma}_0$ in Theorem \ref{theorem:pair_diff}, it suffices to show that, even if $[\gamma_0,\delta_0]^\intercal(Z_1-Z_2)=0$, there is a non-zero variation in $\bar{\nu}_r^\intercal(Z_1-Z_2)$ for every $r=1,\ldots,R_A$. Consider the point $z$ in the assumption of Theorem \ref{theorem:rank_Sigma2}. Since $z_A$ is an interior point, there is a positive constant $c$ such that  $[z_A^\intercal,z_B^\intercal]^\intercal+c\bar{\nu}_r$ belongs to the support of $Z$. Define $z_1=z$ and $z_2=z+c\bar{\nu}_r$. This $z_2$ and $z_1$ are support points of $Z$ such that $[\gamma_0,\delta_0]^\intercal(z_2-z_1)=0$ and $\bar{\nu}_r^\intercal(z_2-z_1)=c\bar{\nu}_r^\intercal\bar{\nu}_r\ne 0$. Finally, note that if $\delta_0^\intercal Z$ is discrete, then $\delta^\intercal_0 Z_A = 0$ and $\mathrm{rank}(Var([\gamma_0,\delta_0]^\intercal[Z_{A}^\intercal,0^\intercal]^\intercal))=\mathrm{rank}(Var(\gamma_0^\intercal[Z_{A}^\intercal,0^\intercal]^\intercal))\le J_1$.
\end{proof} 

\subsection{Proof of Theorem \ref{theorem:rank_Sigma3}}

\begin{proof} 
We use $R_A$ and $(\bar{\nu}_1,\ldots,\bar{\nu}_{R_A})$ in the proof of Theorem \ref{theorem:rank_Sigma2}. There are linearly independent vectors $\bar{\nu}_{R_A+1},\ldots,\bar{\nu}_{R_A+\mathrm{rank}(Var(Z_B))}$ in the support of $[0^\intercal,(Z_{2,B}-Z_{1,B})^\intercal]^\intercal$. Note that the vectors   $\bar{\nu}_1,\ldots,\bar{\nu}_{R_A+\mathrm{rank}(Var(Z_B))}$ are linearly independent. By Theorem \ref{theorem:pair_diff}, it suffices to show that, even if $[\gamma_0,\delta_0]^\intercal(Z_1-Z_2)=0$, there is a non-zero variation in $\bar{\nu}_r^\intercal(Z_1-Z_2)$ for every $r=1,\ldots,R_A+\mathrm{rank}(Var(Z_B))$. The proof for $r=1,\ldots,R_A$ is the same as in the proof of Theorem \ref{theorem:rank_Sigma2}. Consider $r=R_A+1,\ldots,R_A+\mathrm{rank}(Var(Z_B))$. There are $z_{1,B}$ and $z_{2,B}$ in the support of $Z_{B}$ such that 
$$
[0^\intercal,(z_{1,B}-z_{2,B})^\intercal]^\intercal=\bar{\nu}_{r}.
$$
Let $z_{1,A}$ be any point such that $[z_{1,A}^\intercal,z_{1,B}^\intercal]^\intercal$ is in the support of $Z$. By the assumption of this theorem, we can find a point $z_{2,A}$ such that 
$$
[\gamma_0,\delta_0]^\intercal[z_{2,A}^\intercal,z_{2,B}^\intercal]^\intercal=[\gamma_0,\delta_0]^\intercal[z_{1,A}^\intercal,z_{1,B}^\intercal]^\intercal
$$
and $[z_{2,A}^\intercal,z_{2,B}^\intercal]^\intercal$ is in the support of $Z$. Define $z_1=[z_{1,A}^\intercal,z_{1,B}^\intercal]^\intercal$ and $z_2=[z_{2,A}^\intercal,z_{2,B}^\intercal]^\intercal$. This $z_2$ and $z_1$ are support points of $Z$ such that $[\gamma_0,\delta_0]^\intercal(z_2-z_1)=0$ and $\bar{\nu}_r^\intercal(z_2-z_1)=\bar{\nu}_r^\intercal\bar{\nu}_r\ne 0$. To conclude, note $\mathrm{rank}(Var(Z_B))=\dim(Z_B)$.
\end{proof} 

\subsection{Proof of Theorem \ref{theorem:main}}

By Assumption \ref{assn:key_equivalence}, we can reparametrize  the vector $\theta$ such that  
$$
\theta=(\mathrm{vec}(\boldsymbol{\gamma}(\theta))^\intercal,\rho^\intercal)^\intercal
$$
using a finite-dimensional vector $\rho$. For the proof, we assume the above equality with $\theta=(\mathrm{vec}(\gamma)^\intercal,\rho^\intercal)^\intercal$ and $\theta_0=(\mathrm{vec}(\gamma_0)^\intercal,\rho_0^\intercal)^\intercal$. The proof of Theorem \ref{theorem:main} uses the following lemmas. 
  
\begin{lemma}\label{lemma:max_transformed}
$\tilde\theta$ maximizes $\theta\mapsto\hat{Q}(\left[\mathrm{vec}((I-\hat{\Sigma}^\dagger\hat{\Sigma})\gamma)^\intercal,\rho^\intercal\right]^\intercal)$ over $\theta\in\Theta$.
\end{lemma}
\begin{proof}
Let $\theta$ be any element of $\Theta$. Since $\hat{\Sigma}(I-\hat{\Sigma}^\dagger\hat{\Sigma})\gamma=0$, we have $[\mathrm{vec}((I-\hat{\Sigma}^\dagger\hat{\Sigma})\gamma)^\intercal,\rho^\intercal]^\intercal$ satisfies the constraint $\  \hat{\Sigma}\boldsymbol{\gamma}(\theta)=0$. By definition of $\tilde\theta$, we have $\hat{Q}([\mathrm{vec}((I-\hat{\Sigma}^\dagger\hat{\Sigma})\gamma)^\intercal,\rho^\intercal]^\intercal)\leq\hat{Q}(\tilde\theta)=\hat{Q}([\mathrm{vec}(\tilde\gamma)^\intercal,\tilde\rho^\intercal]^\intercal)$. Since $\hat{\Sigma}\tilde\gamma=0$, we have $\hat{Q}([\mathrm{vec}((I-\hat{\Sigma}^\dagger\hat{\Sigma})\gamma)^\intercal,\rho^\intercal]^\intercal)\leq\hat{Q}([\mathrm{vec}((I-\hat{\Sigma}^\dagger\hat{\Sigma})\tilde\gamma)^\intercal,\tilde\rho^\intercal]^\intercal)$.
\end{proof}

\begin{lemma}\label{lemma:sigma_conv}
Under Assumptions \ref{assn:tildeSigma_converge}, $\hat{\Sigma}^\dagger\hat{\Sigma}={\Sigma_0}^\dagger{\Sigma}_0+O_p(1)\|\hat{\Sigma}-{\Sigma_0}\|$.
\end{lemma}
\begin{proof}
With probability approaching one, $\mathrm{rank}(\hat{\Sigma})=\mathrm{rank}({\Sigma_0})$, so by \textcite[Theorem 20.8.3]{Harville1997}, we have the statement of this lemma. 
\end{proof}

\begin{lemma}\label{lemma:hatSigma_conv_inside_hatQ}
Under the assumptions in Theorem \ref{theorem:main}, \\
$\hat{Q}([\mathrm{vec}((I-\hat{\Sigma}^\dagger\hat{\Sigma})\gamma_0)^\intercal,\rho_0^\intercal]^\intercal)-\hat{Q}([\mathrm{vec}((I-{\Sigma_0}^\dagger{\Sigma_0})\gamma_0)^\intercal,\rho_0^\intercal]^\intercal)=o_p(1)$.
\end{lemma}
\begin{proof}
By the mean-value expansion, with probability approaching one,
\begin{align*}
&|\hat{Q}([\mathrm{vec}((I-\hat{\Sigma}^\dagger\hat{\Sigma})\gamma_0)^\intercal,\rho_0^\intercal]^\intercal)-\hat{Q}([\mathrm{vec}((I-{\Sigma_0}^\dagger{\Sigma_0})\gamma_0)^\intercal,\rho_0^\intercal]^\intercal)|
\\&\leq 
2\sup_{\theta\in\mathcal{N}}|\hat{Q}(\theta)-Q_0(\theta)|
+
|Q_0(\left[\mathrm{vec}((I-\hat{\Sigma}^\dagger\hat{\Sigma})\gamma_0)^\intercal,\rho_0^\intercal\right]^\intercal)-Q_0(\left[\mathrm{vec}((I-{\Sigma_0}^\dagger{\Sigma_0})\gamma_0)^\intercal,\rho_0^\intercal\right]^\intercal)|
\\
&\leq 
2\sup_{\theta\in\mathcal{N}}|\hat{Q}(\theta)-Q_0(\theta)|
+
\sup_{\theta\in\Theta}\|Q_0^{(1)}(\theta)\|\|\hat{\Sigma}^\dagger\hat{\Sigma}-{\Sigma_0}^\dagger{\Sigma_0}\|\|\gamma_0\|.
\end{align*}
Lemma \ref{lemma:sigma_conv} and Assumption \ref{assn:unif_converge} imply the statement of this lemma.  
\end{proof}

\begin{lemma}\label{lemma:consiste_tildetheta}
Suppose the assumptions in Theorem \ref{theorem:main}.
(a) $\tilde\theta-\theta_0=o_p(1)$. 
(b) $\tilde\theta$ is in the interior of the compact space $\Theta$ with probability approaching one. 
\end{lemma}
\begin{proof}
Note that 
\begin{eqnarray*}
Q_0(\tilde\theta)-Q_0(\theta_0)
&=&
Q_0(\tilde\theta)-\hat{Q}(\tilde\theta)
\\&&
+\hat{Q}(\left[\mathrm{vec}((I-\hat{\Sigma}^\dagger\hat{\Sigma})\tilde\gamma)^\intercal,\tilde\rho^\intercal\right]^\intercal)-\hat{Q}(\left[\mathrm{vec}((I-\hat{\Sigma}^\dagger\hat{\Sigma})\gamma_0)^\intercal,\rho_0^\intercal\right]^\intercal)
\\&&
+\hat{Q}(\left[\mathrm{vec}((I-\hat{\Sigma}^\dagger\hat{\Sigma})\gamma_0)^\intercal,\rho_0^\intercal\right]^\intercal)-\hat{Q}(\theta_0)
\\&&
+\hat{Q}(\theta_0)-Q_0(\theta_0)
\\
&\geq&
-2\sup_{\theta\in\Theta}|\hat{Q}(\theta)-Q_0(\theta)|+o_p(1)
\end{eqnarray*}
where the equality follows from $\tilde\gamma=(I-\hat{\Sigma}^\dagger\hat{\Sigma})\tilde\gamma$ and the inequality follows from Lemma \ref{lemma:max_transformed} and \ref{lemma:hatSigma_conv_inside_hatQ}. 
Then, by Assumption \ref{assn:unif_converge}, we have $Q_0(\tilde\theta)\geq Q_0(\theta_0)+o_p(1)$. 
By the compactness of $\Theta$ and the uniqueness of $\theta_0$, the first statement of this lemma holds. The second statement follows from the first statement and Assumption \ref{assn:unif_converge}(i).
\end{proof}

\begin{lemma}\label{lemma:Qhat_der_conv}
Under the assumptions in Theorem \ref{theorem:main}, $$
\|\hat{Q}^{(1)}(\tilde\theta)-{Q}_0^{(2)}(\theta_0)(\tilde\theta-\theta_0)\|
\leq
o_p(1)\|\tilde\theta-\theta_0\|+O_p(n^{-1/2}).
$$ 
\end{lemma}
\begin{proof}
Since ${Q}_0^{(1)}(\theta_0)=0$ from the first-order condition for $\theta_0$, we have 
\begin{align*}
\hat{Q}^{(1)}(\tilde\theta)-{Q}_0^{(2)}(\theta_0)(\tilde\theta-\theta_0)
&=
((\hat{Q}^{(1)}(\tilde\theta)-{Q}_0^{(1)}(\tilde\theta))-(\hat{Q}^{(1)}(\theta_0)-{Q}_0^{(1)}(\theta_0)))
\\&\quad+({Q}_0^{(1)}(\tilde\theta)-{Q}_0^{(1)}(\theta_0)-{Q}_0^{(2)}(\theta_0)(\tilde\theta-\theta_0))
\\&\quad+(\hat{Q}^{(1)}(\theta_0)-{Q}_0^{(1)}(\theta_0)).
\end{align*}
The first term $((\hat{Q}^{(1)}(\tilde\theta)-{Q}_0^{(1)}(\tilde\theta))-(\hat{Q}^{(1)}(\theta_0)-{Q}_0^{(1)}(\theta_0)))$ is $o_p(1)\|\tilde\theta-\theta_0\|$ because the mean value theorem and Assumption \ref{assn:unif_converge} imply  
$$
\|((\hat{Q}^{(1)}(\tilde\theta)-{Q}_0^{(1)}(\tilde\theta))-(\hat{Q}^{(1)}(\theta_0)-{Q}_0^{(1)}(\theta_0)))\|
\leq
\sup_{\theta\in\mathcal{N}}\|\hat{Q}^{(2)}(\theta)-{Q}_0^{(2)}(\theta)\|\|\tilde\theta-\theta_0\|.
=
o_p(1)\|\tilde\theta-\theta_0\|.
$$
The second term $({Q}_0^{(1)}(\tilde\theta)-{Q}_0^{(1)}(\theta_0)-{Q}_0^{(2)}(\theta_0)(\tilde\theta-\theta_0))$ is $o_p(1)\|\tilde\theta-\theta_0\|$ because the first-order Taylor expansion and Lemma \ref{lemma:consiste_tildetheta} imply  
$$
\|{Q}_0^{(1)}(\tilde\theta)-{Q}_0^{(1)}(\theta_0)-{Q}_0^{(2)}(\theta_0)(\tilde\theta-\theta_0)\|
\leq o_p(1)\|\tilde\theta-\theta_0\|.
$$
The third term $\hat{Q}^{(1)}(\theta_0)-{Q}_0^{(1)}(\theta_0)$ is $O_p(n^{-1/2})$ by Assumption \ref{assn:unif_converge}. 
\end{proof}

\begin{lemma}\label{lemma:tilde_consist}
Under the assumptions in Theorem \ref{theorem:main}, $\tilde\theta-\theta_0=O_p(1)\max\{\|\hat{\Sigma}-{\Sigma_0}\|,n^{-1/2}\}$.
\end{lemma}
\begin{proof}
By Lemmas \ref{lemma:max_transformed} and \ref{lemma:consiste_tildetheta}(b), the first-order condition for $\tilde\theta$ and constraint may be written as 
$$
\left(\begin{array}{c}
\left.\frac{\partial}{\partial\theta}\hat{Q}\left(\left[\mathrm{vec}((I-\hat{\Sigma}^\dagger\hat{\Sigma})\gamma)^\intercal,\rho^\intercal\right]^\intercal\right)\right|_{\theta=\tilde\theta}\\\hat{\Sigma}\mathrm{vec}(\tilde\gamma)
\end{array}\right)
=0.
$$
Define $$\mathbf{M}_3=
\left(\begin{array}{c}
\left(\begin{array}{cc}I_{J_1}\otimes(I_{\dim(Z)}-{\Sigma_0}^\dagger{\Sigma_0})&O\\ O&I_{\dim(\rho)}\end{array}\right){Q}_0^{(2)}(\theta_0)\\
(\begin{array}{cc}I_{J_1}\otimes{\Sigma_0}^\dagger{\Sigma_0}&O\end{array})
\end{array}\right).$$
By ${\Sigma_0}\theta_0=0$ and Lemmas \ref{lemma:sigma_conv} and \ref{lemma:consiste_tildetheta}, we have  
$$
\mathbf{M}_3
(\tilde\theta-\theta_0)
=
O(1)(
\hat{Q}^{(1)}(\tilde\theta)-{Q}_0^{(2)}(\theta_0)(\tilde\theta-\theta_0))
+
O_p(1)\|\hat{\Sigma}-{\Sigma_0}\|.
$$
Note that $\mathbf{M}_3$ has full column rank, because  
$$
\mathrm{rank}\left(\mathbf{M}_3\right)
=
\mathrm{rank}
\left(\begin{array}{cc}I_{J_1\dim(Z)}-I_{J_1}\otimes{\Sigma_0}^\dagger{\Sigma_0}&O\\ O&I_{\dim(\rho)}\\I_{J_1}\otimes{\Sigma_0}^\dagger{\Sigma_0}&O\end{array}\right)=\dim(\theta).
$$
Therefore, 
$$
\tilde\theta-\theta_0
=
O(1)(
\hat{Q}^{(1)}(\tilde\theta)-{Q}_0^{(2)}(\theta_0)(\tilde\theta-\theta_0))
+
O_p(1)\|\hat{\Sigma}-{\Sigma_0}\|.
$$
By Lemma \ref{lemma:Qhat_der_conv}, 
$$
\|\tilde\theta-\theta_0\|
\leq
o_p(1)\|\tilde\theta-\theta_0\|+
O_p(1)\max\{\|\hat{\Sigma}-{\Sigma_0}\|,n^{-1/2}\},
$$
which implies the statement of this lemma holds. 
\end{proof}
 
\begin{proof}[Proof of Theorem \ref{theorem:main}] 
By \textcite[Theorem 2.1 and 3.1]{newey1994large} and Assumption \ref{assn:unif_converge}, $\hat{\theta}^*=\theta_0+O_p(n^{-1/2})$. By Lemma \ref{lemma:tilde_consist}, 
$$
\tilde\theta-\theta_0=o_p(1)\mbox{ and }\|\tilde\theta-\hat{\theta}^*\|=O_p(1)\max\{\|\hat{\Sigma}-{\Sigma_0}\|,n^{-1/2}\}.
$$ 
Thus the statement of this theorem follows from \textcite[Theorem 2]{robinson1988stochastic}. Assumption  A1 in \textcite{robinson1988stochastic} follows from Assumption \ref{assn:unif_converge} and the consistency of $\hat{\theta}^*$. Assumption  A3 in \textcite{robinson1988stochastic} follows from Assumption \ref{assn:unif_converge}.
\end{proof}

\subsection{Proof of Lemma  \ref{lemma:Sigma_estimation}}
\begin{proof}
Since $\tilde{\Sigma}={\Sigma_0}+o_p(1)$, it suffices to show ${\mathrm{rank}}(\hat{\Sigma})={\mathrm{rank}}({\Sigma}_0)$ and $\|\hat{\Sigma}-{\Sigma_0}\|\leq 2\|\tilde{\Sigma}-{\Sigma_0}\|$.
By the assumption of this lemma, we have 
$Pr(\|\tilde{\Sigma}-{\Sigma}_0\|\leq\kappa\leq\min\{\lambda_k:\lambda_k>0\}-\|\tilde{\Sigma}-{\Sigma}_0\|)=1+o(1)$, where $\lambda_1\geq\cdots\geq\lambda_K$ are the eigenvalues of ${\Sigma}_0$.
As long as $\|\tilde{\Sigma}-{\Sigma}_0\|\leq\kappa\leq\min\{\lambda_k:\lambda_k>0\}-\|\tilde{\Sigma}-{\Sigma}_0\|$, by Weyl's inequality on the eigenvalue perturbations, we have 
$$
1\{\hat\lambda_k>\kappa\}=1\{\lambda_k>0\}
$$
for every $k=1,\ldots,K$. 
It implies ${\mathrm{rank}}(\hat{\Sigma})={\mathrm{rank}}({\Sigma}_0)$ with probability approaching one. Moreover, by the Eckart-Young-Mirsky theorem, $
\|\hat{\Sigma}-\tilde{\Sigma}\|\leq\|{\Sigma}_0-\tilde{\Sigma}\|$, which implies $\|\hat{\Sigma}-{\Sigma_0}\|\leq\|\hat{\Sigma}-\tilde{\Sigma}\|+\|\tilde{\Sigma}-{\Sigma_0}\|\leq 2\|\tilde{\Sigma}-{\Sigma_0}\|$.
\end{proof}

\section{Estimation of the constraint matrix}\label{sec:const_tildeSigma}

In this section we propose a consistent estimator for $\Sigma_{0}$ from an $n$ i.i.d. observations $Z_1,\ldots,Z_n$ and an estimator for $(\Pi_0,\delta_0)$. Our construction is related to the estimator of \textcite{ahn2018simple}.

As is relevant for dynamic discrete choice models, we allow some components of $Z$ to be discrete. In this section, we arrange $[\Pi_0(Z)^\intercal,(\delta_0^\intercal Z)^\intercal]^\intercal$ and write it as $[U^\intercal,V^\intercal]^\intercal$, where $U$ is a continuous random variable and $V$ is a random variable with finite support. With some abuse of notation, we use $\left[\Pi_0(Z)^\intercal,(\delta_0^\intercal Z)^\intercal\right]^\intercal$ and $[U^\intercal,V^\intercal]^\intercal$ interchangeably. The proposed estimator uses kernel smoothing, and therefore we require conditions on both the kernel function $\mathbf{K}$ and the bandwidth $h$:
\begin{assumption}\label{assn:kernl+band}
(i) $\mathbf{K}:\mathbb{R}^{\dim(U)+\dim(V)}\rightarrow\mathbb{R}$ has a bounded first derivative $\mathbf{K}^{(1)}$. 
(ii) $\mathbf{K}\left(\left[u^\intercal,v^\intercal\right]^\intercal\right)=0$ for every $(u,v)$ with $\|u\|\geq 1$ and $v\ne 0$.  
(iii) $\int \mathbf{K}\left(\left[u^\intercal,0^\intercal\right]^\intercal\right)du=1$ and $\int \mathbf{K}\left(\left[u^\intercal,0^\intercal\right]^\intercal\right)udu=0$.
(iv) $h\rightarrow 0$ and $nh^{\dim(U)/2}\rightarrow\infty$ as $n\rightarrow\infty$. 
\end{assumption} 

To construct an estimator for $\Sigma_{0}$, we assume that there is a consistent estimator $(\hat\delta,\hat\Pi)$ for $(\delta_0,\Pi_0)$. As in Section \ref{subsec:ddcm}, in dynamic discrete models $\delta_{0}$ may govern the state transition kernel, and is thus consistently estimable from data on the state transition. Similarly, the CCPs $\Pi_{0}$ are nonparametrically identified from the data.

\begin{assumption}\label{assn:prel_estimation}
$\max\{\sup_{z}\|\hat\Pi(z)-\Pi_0(z))\|,\|\hat\delta-\delta_0\|\}=o_p(h)$.
\end{assumption}

\begin{assumption}\label{assn:smoothness_regrefun}
(i) The functions $E[(Z_1-Z_2)(Z_1-Z_2)^\intercal\mid U_1-U_2=\cdot,V_1=V_2]f_{U_1-U_2\mid V_1=V_2}(\cdot)$ and $f_{U_1-U_2\mid V_1=V_2}(\cdot)$ are twice continuously differentiable near zero.
(ii) $f_{U_1-U_2}$, $f_{U_1-U_2\mid Z_1}$, $E\left[\|Z_2\|\mid U_1-U_2,Z_1\right]$, $E\left[\|Z_2\|^2\mid U_1-U_2,Z_1\right]$, and $E[\|Z_1-Z_2\|^4\mid U_1-U_2,V_1=V_2]$ are bounded. 
\end{assumption}

With these assumptions in hand, we define 
\begin{equation}\label{eq:tilde_sigma}
\tilde{\Sigma}\equiv \dfrac{\sum_{i_1,i_2}\mathbf{K}\left([(\hat\Pi(Z_{i_1})-\hat\Pi(Z_{i_2}))^\intercal,(\hat\delta^\intercal(Z_{i_1}-Z_{i_2}))^\intercal]^\intercal/h\right)(Z_{i_1}-Z_{i_2})(Z_{i_1}-Z_{i_2})^\intercal}{\sum_{i_1,i_2}\mathbf{K}\left([(\hat\Pi(Z_{i_1})-\hat\Pi(Z_{i_2}))^\intercal,(\hat\delta^\intercal(Z_{i_1}-Z_{i_2}))^\intercal]^\intercal/h\right)}.
\end{equation}
The following result shows the consistency for $\tilde{\Sigma}$.

\begin{theorem}\label{thm:tilde_sigma}
    If $Z_1,\ldots, Z_n$ are i.i.d. and Assumptions \ref{assn:kernl+band}-\ref{assn:smoothness_regrefun}  hold, then $\tilde{\Sigma}-\Sigma_{0}=o_{p}(1)$
\end{theorem}
\subsection{Proof of Theorem \ref{thm:tilde_sigma}} 
We use the following lemmas to prove this theorem. Define 
$\zeta_i\equiv[\Pi_0(Z_i)^\intercal,(\delta_0^\intercal Z_i)^\intercal]^\intercal$ and $\hat\zeta_i\equiv[\hat\Pi(Z_i)^\intercal,(\hat\delta^\intercal Z_i)^\intercal]^\intercal$. Define $\hat{W}_{i_1i_2}\equiv\frac{1}{h^{\dim(U)}}\mathbf{K}\left((\hat\zeta_{i_1}-\hat\zeta_{i_2})/h\right)(Z_{i_1}-Z_{i_2})(Z_{i_1}-Z_{i_2})^\intercal$ and ${W}_{i_1i_2}\equiv\frac{1}{h^{\dim(U)}}\mathbf{K}\left((\zeta_{i_1}-\zeta_{i_2})/h\right)(Z_{i_1}-Z_{i_2})(Z_{i_1}-Z_{i_2})^\intercal$.
\begin{lemma}\label{lemma_num1}
Under the assumptions in Theorem \ref{thm:tilde_sigma},
$\frac{1}{n^2}\sum_{i_1,i_2}(\hat{W}_{i_1i_2}-{W}_{i_1i_2})=o_p(1)$.
\end{lemma}

\begin{proof}
By Assumption \ref{assn:prel_estimation}, we can assume $\|(\hat\zeta_{i_1}-\hat\zeta_{i_2})-(\zeta_{i_1}-\zeta_{i_2})\|<h$ without loss of generality. 
Thus 
$\|\zeta_{i_1}-\zeta_{i_2}\|\geq 2h\implies \|\hat\zeta_{i_1}-\hat\zeta_{i_2}\|\geq h$, and therefore 
$$
\left|\mathbf{K}\left((\hat\zeta_{i_1}-\hat\zeta_{i_2})/h\right)-
\mathbf{K}\left((\zeta_{i_1}-\zeta_{i_2})/h\right)\right|
\leq
1\{\|\zeta_{i_1}-\zeta_{i_2}\|\leq 2h\}\left|\mathbf{K}\left((\hat\zeta_{i_1}-\hat\zeta_{i_2})/h\right)-
\mathbf{K}\left((\zeta_{i_1}-\zeta_{i_2})/h\right)\right|.
$$
By the second-order Taylor expansion, there is some constant $C$ such that 
$$
\mathbf{K}\left((\hat\zeta_{i_1}-\hat\zeta_{i_2})/h\right)-
\mathbf{K}\left((\zeta_{i_1}-\zeta_{i_2})/h\right)
=\frac{1}{h}\mathbf{K}^{(1)}\left((\zeta_{i_1}-\zeta_{i_2})/h\right)((\hat\zeta_{i_1}-\hat\zeta_{i_2})-(\zeta_{i_1}-\zeta_{i_2}))
+\frac{1}{h^2}R_{2,i_1i_2}
$$
with $\|R_{2,i_1i_2}\|\leq C\left\|(\hat\zeta_{i_1}-\hat\zeta_{i_2})-(\zeta_{i_1}-\zeta_{i_2})
\right\|^2$. 
Therefore, 
\begin{align*}
&\left|\mathbf{K}\left((\hat\zeta_{i_1}-\hat\zeta_{i_2})/h\right)-
\mathbf{K}\left((\zeta_{i_1}-\zeta_{i_2})/h\right)\right|
\\&\leq
\frac{1}{h}\left\|\mathbf{K}^{(1)}\left((\zeta_{i_1}-\zeta_{i_2})/h\right)\right\|\|(\hat\zeta_{i_1}-\hat\zeta_{i_2})-(\zeta_{i_1}-\zeta_{i_2})\|+\frac{1}{h^2}1\{\|\zeta_{i_1}-\zeta_{i_2}\|\leq 2h\}\left\|R_{2,i_1i_2}\right\|.
    \end{align*}
Since $\left\|\hat{W}_{i_1i_2}-{W}_{i_1i_2}\right\|
\leq
\frac{1}{h^{\dim(U)}}\left|\mathbf{K}\left((\hat\zeta_{i_1}-\hat\zeta_{i_2})/h\right)-
\mathbf{K}\left((\zeta_{i_1}-\zeta_{i_2})/h\right)\right|\|Z_{i_1}-Z_{i_2}\|^2$,
we have 
\begin{eqnarray*}
&&
\left\|\frac{1}{n^2}\sum_{i_1,i_2}\left(\hat{W}_{i_1i_2}-{W}_{i_1i_2}\right)\right\|
\\
&&\leq
\frac{1}{n^2}\sum_{i_1,i_2}\frac{1}{h^{\dim(U)+1}}\left\|\mathbf{K}^{(1)}\left((\zeta_{i_1}-\zeta_{i_2})/h\right)\right\|
\|Z_{i_1}-Z_{i_2}\|^2\left\|(\hat\zeta_{i_1}-\hat\zeta_{i_2})-(\zeta_{i_1}-\zeta_{i_2})\right\|
\\&&\quad+
\frac{1}{n^2}\sum_{i_1,i_2}\frac{1}{h^{\dim(U)+2}}
1\{\|\zeta_{i_1}-\zeta_{i_2}\|\leq 2h\}\|Z_{i_1}-Z_{i_2}\|^2\left\|R_{2,i_1i_2}\right\|
\\
&&\leq
\mathcal{U}_1\frac{1}{h}\sup_{(i_1,t_1,i_2,t_2): i_1\ne i_2}\left\|(\hat\zeta_{i_1}-\hat\zeta_{i_2})-(\zeta_{i_1}-\zeta_{i_2})\right\|
+C\mathcal{U}_2\frac{1}{h^{2}}
\sup_{(i_1,t_1,i_2,t_2): i_1\ne i_2}\left\|(\hat\zeta_{i_1}-\hat\zeta_{i_2})-(\zeta_{i_1}-\zeta_{i_2})
\right\|^2,
\end{eqnarray*}
where 
$\mathcal{U}_1\equiv\frac{1}{n^2}\sum_{i_1,i_2}\frac{1}{h^{\dim(U)}}\|\mathbf{K}^{(1)}((\zeta_{i_1}-\zeta_{i_2})/h)\|\|Z_{i_1}-Z_{i_2}\|^2$ and $\mathcal{U}_2\equiv\frac{1}{n^2}\sum_{i_1,i_2}\frac{1}{h^{\dim(U)}}1\{\|\zeta_{i_1}-\zeta_{i_2}\|\leq 2h\}\|Z_{i_1}-Z_{i_2}\|^2$.
To show this lemma, by Assumption \ref{assn:prel_estimation}, it suffices to show $\mathcal{U}_1=O_p(1)$
 and $\mathcal{U}_2=O_p(1)$.
Note that 
\begin{eqnarray*}
E[|\mathcal{U}_1|]
&\leq&
\frac{1}{n^2}\sum_{i_1,i_2}E[\frac{1}{h^{\dim(U)}}\left\|\mathbf{K}^{(1)}\left((\zeta_{i_1}-\zeta_{i_2})/h\right)\right\|
E[\|Z_{i_1}-Z_{i_2}\|^2\mid \zeta_{i_1}-\zeta_{i_2}]]\\
&\leq&
CE\left[\frac{1}{h^{\dim(U)}}\left\|\mathbf{K}^{(1)}\left([(U_1-U_2)^\intercal,(V_1-V_2)^\intercal]^\intercal/h\right)\right\|
\right]
\end{eqnarray*}
for some constant $C$. 
For sufficiently small $h$, 
\begin{eqnarray*}
E[|\mathcal{U}_1|]
&\leq&
CE\left[\frac{1}{h^{\dim(U)}}\left\|\mathbf{K}^{(1)}\left([(U_1-U_2)^\intercal/h,0^\intercal]^\intercal\right)\right\|
\right].
\end{eqnarray*}
Using the change of variables, 
\begin{eqnarray*}
E[|\mathcal{U}_1|]
&\leq&
C\int\frac{1}{h^{\dim(U)}}\left\|\mathbf{K}^{(1)}\left([u^\intercal/h,0^\intercal]^\intercal\right)\right\|f_{U_1-U_2}(u)du\\
&=&
C\int \left\|\mathbf{K}^{(1)}\left([u^\intercal,0^\intercal]^\intercal\right)\right\|f_{U_1-U_2}(uh)du\\
&=&
O(1).
\end{eqnarray*}
Similarly, we can show  $\mathcal{U}_2=O_p(1)$. 
\end{proof}

\begin{lemma}\label{lemma_num2}
Under the assumptions in Theorem \ref{thm:tilde_sigma},
$\frac{1}{n^2}\sum_{i_1,i_2}({W}_{i_1i_2}-E[{W}_{i_1i_2}])=o_p(1).$
\end{lemma}
\begin{proof}
Based on the variance formula for U-statistics, it suffices to show $Var(\frac{1}{h^{\dim(U)}}\mathbf{K}(\zeta_{12}/h)\|Z_{1}-Z_{2}\|^2)=O(h^{-\dim(U)})$ and $Var(E[\frac{1}{h^{\dim(U)}}\mathbf{K}(\zeta_{12}/h)\|Z_{1}-Z_{2}\|^2\mid Z_{1}])=O(1)$.

First, we are going to show $Var\left(\frac{1}{h^{\dim(U)}}\mathbf{K}\left(\zeta_{12}/h\right)
\|Z_{1}-Z_{2}\|^2\right)=O(h^{-\dim(U)})$. 
Note that 
\begin{eqnarray*}
Var\left(\frac{1}{h^{\dim(U)}}\mathbf{K}\left(\zeta_{12}/h\right)
\|Z_{1}-Z_{2}\|^2\right)
&\leq&
E\left[\frac{1}{h^{2\dim(U)}}\mathbf{K}\left(\zeta_{12}/h\right)^2
E\left[\|Z_{1}-Z_{2}\|^4\mid\zeta_{12}\right]\right]
\\
&=&
O(1)
E\left[\frac{1}{h^{2\dim(U)}}\mathbf{K}\left(\zeta_{12}/h\right)^2
\right].
\end{eqnarray*}
For sufficiently small $h$, 
$$
Var\left(\frac{1}{h^{\dim(U)}}\mathbf{K}\left(\zeta_{12}/h\right)
\|Z_{1}-Z_{2}\|^2\right)
=O(1)
E\left[\frac{1}{h^{2\dim(U)}}\mathbf{K}\left(\left[(U_{1}-U_{2})^\intercal/h,0^\intercal\right]^\intercal\right)^2.
\right]
$$
Using the change of variables, 
\begin{eqnarray*}
Var\left(\frac{1}{h^{\dim(U)}}\mathbf{K}\left(\zeta_{12}/h\right)
\|Z_{1}-Z_{2}\|^2\right)
&=&
O(1)
\int \frac{1}{h^{2\dim(U)}}\mathbf{K}\left(\left[u^\intercal/h,0^\intercal\right]^\intercal\right)^2
f_{U_1-U_2}(u)du
\\
&=&
O(1)
\int \frac{1}{h^{\dim(U)}}\mathbf{K}\left(\left[u^\intercal,0^\intercal\right]^\intercal\right)^2
f_{U_1-U_2}(uh)du
\\
&=&
O(h^{-(\dim(U))}).
\end{eqnarray*}

Second, we are going to show $Var\left(E\left[\frac{1}{h^{\dim(U)}}\mathbf{K}\left(\zeta_{12}/h\right)
\|Z_{1}-Z_{2}\|^2\mid Z_{1}\right]\right)=O(1)$. 
For sufficiently small $h$, 
$$
E\left[\mathbf{K}\left(\zeta_{12}/h\right)
\|Z_{1}-Z_{2}\|^2\mid Z_{1}\right]\leq E\left[\mathbf{K}\left(\left[(U_{1}-U_{2})^\intercal/h,0^\intercal\right]^\intercal\right)
\|Z_{1}-Z_{2}\|^2\mid Z_{1}\right].
$$
Since $E\left[\|Z_{2}\|^2\mid U_{1}-U_{2},Z_{1}\right]$ and $E\left[\|Z_{2}\|^2\mid U_{1}-U_{2},Z_{1}\right]$ are bounded, there are some constants $C_0,C_1,C_2$ such that 
$$
E\left[\mathbf{K}\left(\zeta_{12}/h\right)
\|Z_{1}-Z_{2}\|^2\mid Z_{1}\right]\leq E\left[\mathbf{K}\left(\left[(U_{1}-U_{2})^\intercal/h,0^\intercal\right]^\intercal\right)
(C_0+C_1\|Z_{1}\|+C_2\|Z_{1}\|^2)
\mid Z_{1}\right].
$$
Using the change of variables, 
\begin{eqnarray*}
E\left[\mathbf{K}\left(\zeta_{12}/h\right)
\|Z_{1}-Z_{2}\|^2\mid Z_{1}\right]
&\leq&
\int\mathbf{K}\left(\left[u^\intercal/h,0^\intercal\right]^\intercal\right)
f_{U_{1}-U_{2}\mid Z_{1}}(u)du
(C_0+C_1\|Z_{1}\|+C_2\|Z_{1}\|^2)
\\
&=&
h^{\dim(U)}\int \mathbf{K}\left(\left[u^\intercal,0^\intercal\right]^\intercal\right)
f_{U_{1}-U_{2}\mid Z_{1}}(uh)du
(C_0+C_1\|Z_{1}\|+C_2\|Z_{1}\|^2).
\end{eqnarray*}
Therefore, 
$Var\left(E\left[\frac{1}{h^{\dim(U)}}\mathbf{K}\left(\zeta_{12}/h\right)
\|Z_{1}-Z_{2}\|^2\mid Z_{1}\right]\right)=O(1)$. 
\end{proof}

\begin{lemma}\label{lemma_num3}
Under the assumptions in Theorem \ref{thm:tilde_sigma},
$\frac{1}{n^2}\sum_{i_1,i_2}E[{W}_{i_1i_2}]={\Sigma}_0Pr(V_1=V_2)f_{U_1-U_2\mid V_1=V_2}(0)+o_p(1)$.
\end{lemma}
\begin{proof}
By Assumption \ref{assn:kernl+band}, for sufficiently small $h$, we have 
$$
E[{W}_{12}]
=
E\left[\frac{1}{h^{\dim(U)}}\mathbf{K}\left(\left[(U_1-U_2)^\intercal/h,0^\intercal\right]^\intercal\right)(Z_1-Z_2)(Z_1-Z_2)^\intercal1\{ V_1=V_2\}\right].
$$
It suffices to show $E[{W}_{12}\mid V_1=V_2]={\Sigma}_0f_{U_1-U_2\mid V_1=V_2}(0)+O( h^2)$. 
Note that 
$$
E[{W}_{12}\mid V_1=V_2]=E\left[\frac{1}{h^{\dim(U)}}\mathbf{K}\left(\left[(U_1-U_2)^\intercal/h,0^\intercal\right]^\intercal\right)(Z_1-Z_2)(Z_1-Z_2)^\intercal\mid V_1=V_2\right].
$$
Using the law of iterated expectations and the change of variables, we have 
$$
E[{W}_{12}\mid V_1=V_2]=\int \mathbf{K}\left(\left[u^\intercal,0^\intercal\right]^\intercal\right)E[(Z_1-Z_2)(Z_1-Z_2)^\intercal\mid U_1-U_2=uh,V_1=V_2]f_{U_1-U_2\mid V_1=V_2}(uh)du.
$$
By Assumptions \ref{assn:kernl+band} and \ref{assn:smoothness_regrefun}, the conclusion of this lemma holds. 
\end{proof}

\begin{proof}[Proof of Theorem \ref{thm:tilde_sigma}]
By Lemmas \ref{lemma_num1}, \ref{lemma_num2}, and \ref{lemma_num3},
$$
\frac{1}{n^2}\sum_{i_1,i_2}\hat{W}_{i_1i_2}={\Sigma}_0Pr(V_1=V_2)f_{U_1-U_2\mid V_1=V_2}(0)+o_p(1).
$$
For the denominator, in a similar fashion, we can show that 
$$
\frac{1}{n^2}\sum_{i_1,i_2}\frac{1}{h^{\dim(U)}}\mathbf{K}\left((\hat\zeta_{i_1}-\hat\zeta_{i_2})/h\right)=Pr(V_1=V_2)f_{U_1-U_2\mid V_1=V_2}(0)+o_p(1).
$$
Combining these arguments, we have $\hat{\Sigma}={\Sigma_0}+o_p(1)$.
\end{proof}

\section{Appendix to Section \ref{sec:appl}}\label{app:mcmc}

In this appendix we provide additional details on the Monte Carlo design and our implementation.

\subsection{Estimation and imposition of the constraint matrix}\label{app:mcmcsigma}

In Section \ref{sec:appl-model} we described the form of the constraint matrix $\Sigma_0$ within the model of that section. We now describe the estimator for $\Sigma_{0}$ that we adopt for our simulation exercise. First, we estimate $\Pi_{0}(z)$ by the Nadaraya–Watson kernel estimator for $\Pr(N_{t+1}=0\mid W=w)$, where the kernel function is the product of one-dimensional Gaussian kernels and the common bandwidth is chosen by leave-one-out cross-validation. Then, noting that $\delta_{0}$ is known in this model, we construct a consistent estimator $\tilde{\Sigma}$ as in equation \eqref{eq:tilde_sigma}. As described around that equation, this construction is based on \textcite{ahn2018simple}. We choose the biweight kernel function and use the rule-of-thumb bandwidth $1.06(n(n-1)T(T-1))^{-1/5}$ (after scaling the smoothed CCP differences to have unit variance). To construct the low-rank approximation $\hat\Sigma$, as in \textcite{ahn2018simple}, we use \textit{a priori} knowledge to set the rank-deficiency. Namely, by the discussion in Section \ref{sec:appl-model}, $\mathrm{rank}(\Sigma_0)\ge 8$, so we set the rank-deficiency of $\Sigma_{0}$ as $2=\dim(Z_t)-8$.

We now describe a systematic approach to implementing the estimated constraints $\hat\Sigma\boldsymbol{\gamma}(\theta)=0$, which we adopt in our Monte Carlo experiment. As indicated in the main paper, because $\hat\Sigma\boldsymbol{\gamma}(\theta)=0$ is linear in the optimizing variable $\theta$, given $\hat\Sigma$, the constraint is simple to impose via a reparameterization of $\gamma=\boldsymbol{\gamma}(\theta)$. Namely, set $\gamma=(\gamma_{fixed}^{\intercal},\gamma_{short}^{\intercal})^{\intercal}$
where $\gamma_{short}\in\mathbb{R}^{\mathrm{rank} (\hat\Sigma)}$ is an unknown parameter to be estimated
and $\gamma_{fixed}\in\mathbb{R}^{\dim(Z)-\mathrm{rank}(\hat\Sigma)}$ is a
known linear function of $\gamma_{short}$. Specifically, by writing $\hat\Sigma=QR$ where $Q$ is orthonormal and $R$'s leading principal submatrix of size $\mathrm{rank}(\hat\Sigma)\times\mathrm{rank}(\hat\Sigma)$ is upper triangular, one has $\gamma_{fixed}=M\gamma_{short}$ for
$$
M=\left(\begin{pmatrix}
    I_{\mathrm{rank}(\hat\Sigma)} \\ 0
\end{pmatrix}^\intercal R \begin{pmatrix}
    I_{\mathrm{rank}(\hat\Sigma)} \\ 0
\end{pmatrix}\right)^{-1} \left(\begin{pmatrix} 0 \\
    I_{\dim(Z) - \mathrm{rank}(\hat\Sigma)} 
\end{pmatrix}^\intercal R \begin{pmatrix} 0 \\
    I_{\dim(Z) - \mathrm{rank}(\hat\Sigma)} 
\end{pmatrix}\right),
$$
so the constrained optimization problem can be solved as an unconstrained optimization problem using the reparameterized objective function. To conclude, we reiterate that the above reparameterization technique is not specific to our model, and can therefore be applied generically. We also note that the constraints can be imposed \textit{outside} the optimization algorithm, in the sense that the algorithm will be searching over $\gamma_{short}\in\mathbb{R}^{\dim(\hat\Sigma)}$.

\subsection{The target estimators}\label{app:mcmc-est}

In this section we provide details about the two target estimators based on \textcite{arcidiacono2011conditional} and \textcite{bajari2011simple}.\footnote{A referee suggests that, in the special case of time-invariant $W_i$, a different approach to reducing the computational cost may be available by interpolating model solutions across $\gamma_W^\intercal W_i$. We leave this direction to future research.}

\subsubsection{Estimator based upon \texorpdfstring{\textcite{arcidiacono2011conditional}}{Arcidiacono and Miller (2011)}}

Applied to the model of Section \ref{sec:appl}, \textcite{arcidiacono2011conditional}'s so-called EM algorithm parameterizes $F_\lambda$ as $\{v,\mu\}=\{(v_1,v_2,\ldots,v_R), (\mu_1,\mu_2,\ldots,\mu_{R-1})\}$ where $R$ is known, but $\mathrm{Supp}(\lambda)=\{v_1,v_2,\ldots,v_R\}$ and $\Pr(\lambda=v_r)=\mu_r$ for $r=1,\ldots,R-1$ are unknown. Then, our implementation of the estimator of \textcite{arcidiacono2011conditional} is defined as
\begin{equation}\label{eq:hatq2}
\{\hat\theta^*_{EM},\hat{v},\hat{\mu}\}=\argmax_{\theta,v,\mu}\sum_{i=1}^{n}\log \sum_{r=1}^R  \mu_r\prod_{t=1}^8 \ell_{i,t}(v_r,\theta).
\end{equation}

To implement the estimator, one must choose $R$, the number of support points of $\lambda$. To ensure correct specification, we set $R=2$. Therefore, imposing that $v$ is strictly increasing, the true values of these parameters are $v_0=(0.1,1.1)^\intercal$, $\mu_0=0.37$.

The remaining tuning parameter for the EM algorithm is the stopping criterion. We continue the EM steps until the average (over the parameter vector) percent change in the absolute value of the parameter is less than 0.025\%.

\subsubsection{Estimator based upon \texorpdfstring{\textcite{bajari2011simple}}{Fox, Kim, Ryan and Bajari (2011)}}

Here we outline our target estimator $\hat\theta^*_{H}$ which is based on the so-called `histogram method' as described in \textcite{bajari2011simple} and \textcite{fox2016simple}. The target estimator parameterizes $F_\lambda$ as $\{v,\mu\}=\{(v_1,v_2,\ldots,v_R), (\mu_1,\mu_2,\ldots,\mu_{R-1})\}$ where $R$ and $\mathrm{Supp}(\lambda)\subseteq\{v_1,v_2,\ldots,v_R\}$ are known, but $\Pr(\lambda=v_r)=\mu_r$ for $r=1,\ldots,R-1$ are unknown. Then, our implementation of this method is defined as
\begin{equation}\label{eq:hatq3}
\{\hat\theta^*_{H},\hat{\mu}\}=\argmax_{\theta,\mu}\sum_{i=1}^{n}\log \sum_{r=1}^R \mu_r \prod_{t=1}^8 \ell_{i,t}(v_r,\theta).
\end{equation}
Relative to the problem in equation \eqref{eq:hatq2}, the vector $v$ is taken to be known. As pointed out in \textcite{bajari2011simple} and elsewhere, the `inner' problem of finding
\begin{equation*}
\mu(\theta)=\argmax_{\mu}\sum_{i=1}^{n}\log \sum_{r=1}^R  \mu_r\prod_{t=1}^8 \ell_{i,t}(v_r,\theta)
\end{equation*}
is a convex programming problem that can be solved efficiently even for very large $R$. As indicated by \textcite{fox2016simple} Section 5.1, this fact motivates the following profiling approach to find $\{\hat\theta_H^*,\hat\mu\}$ as $\hat\mu=\mu(\hat\theta^*_{H})$ and
\begin{equation}\label{eq:hatq4}
\hat\theta^*_{H}=\argmax_{\theta}\sum_{i=1}^{n}\log \sum_{r=1}^R\mu_r(\theta) \prod_{t=1}^8 \ell_{i,t}(v_r,\theta) .
\end{equation}

To implement the estimator, one must choose the vector $v$. We explain our choice of $v$ in Section \ref{ssec:appl_est}, which was chosen to contain the support of $\lambda$, ensuring that the model is correctly specified. We do not vary $v$ or $R$ with sample size.

\subsection{Newton Raphson updates}

Step 2 of our estimator is a sequence of exact Newton-Raphson iterates \parencite[i.e.,][Theorem 2]{robinson1988stochastic} from the preliminary estimator $\tilde\theta$ towards the maximizer of $\hat{Q}(\theta)$. Thus, to implement Step 2 in our Monte Carlo design, we must compute the first and second order derivatives of the objective functions in equations \eqref{eq:hatq2} and \eqref{eq:hatq4}. With the exception of $\mu(\theta)$ (since it is defined as a maximizer), these functions can be differentiated using standard automatic differentiation methods. In our implementation, we use the Julia package Zygote. To differentiate $\mu(\theta)$, we use the results of \textcite[Section B.4.1]{NBERw32164}. For the smallest sample size (i.e., $n=100$) when targeting the EM estimator, we find that the Newton-Raphson iterates sometimes diverge such that the estimated $\Pr(\lambda_i=v)>1$ for some $v$. In this case, we perturb the startup value $\tilde\theta_{EM}$ and reinitialize the Newton-Raphson iterates. Having to find and reinitialize the iterates explains the disproportionate computation time for the $\hat\theta_{EM}$, $n=100$ entry in Table \ref{tab:runtime1}.

\subsection{Starting values}\label{app:startval}

As described in the main paper, our starting values are randomly chosen from $\times_{d}^{D}\{e_{d}-5\Delta,e_{d}-4\Delta,\ldots,e_{d}+4\Delta,e_{d}+5\Delta\}$, which contains $11^{D}$ points, where $D$ is the dimension of the optimization problem. We set $\Delta=1$. The center point of the grid $(e_{1},e_{2},\ldots,e_{D})^{\intercal}$ is chosen as follows. First, we solve
\begin{equation*}
\{\hat{e},\hat{v}\}=\argmax_{\theta,v}\sum_{i=1}^{n}\log \prod_{t=1}^8 \ell_{i,t}(v,\theta).
\end{equation*}
This is equivalent to parametric pseudo maximum likelihood estimation (that is, if $\lambda$ is presumed to be degenerate). For the unconstrained estimator $\hat\theta^*$ (for which $D=\dim(\hat{e})$), we set $e=\hat{e}$. For the constrained estimator $\tilde\theta$, so that the starting values satisfy the imposed constraint, we set $e$ to be the projection of $\hat{e}$ on the nullspace of $\hat\Sigma$. To the set of randomly chosen starting values, we add the center point of the grid.

The EM algorithm of \textcite{arcidiacono2011conditional} additionally requires initialization of the support of $\lambda$. Given $|\mathrm{Supp}(\lambda)|=2$, we set these as $\hat{v}\pm0.5$.

\subsection{General advice for implementation}

For a given structural model satisfying index invertibility and a programmed estimator, we suggest the following steps to implement our method:

\begin{enumerate}
\item Estimate $\Sigma_{0}$, possibly using our proposed estimator explained in Sections \ref{sec:const_tildeSigma} and \ref{app:mcmcsigma}.
\item Given the programmed target estimator, reparameterize the programmed objective function using the QR decomposition of $\hat\Sigma$ given in Section \ref{app:mcmcsigma}. Solve the reparameterized optimization problem to form the Step 1 estimator $\tilde\theta$.
\item Given the programmed target objective function, form first and second order derivatives and perform $L$ Newton-Raphson iterates.
\end{enumerate}

\subsection{Additional Monte Carlo results}\label{sec:additional_MC}

Table \ref{tab:conv2} displays the empirical root mean squared error times $\sqrt{n}$ of the difference between our first step estimator $\tilde\theta$ and the target estimator $\hat\theta^*$ for each approach.

\begin{table}[h]
\centering
\begin{tabular}{l|rrrr|rrrr}
  \hline
   \rule{0pt}{2.5ex}  & \multicolumn{4}{c|}{$\sqrt{n}(\tilde\theta_{H}-\hat\theta^*_{H})$} & \multicolumn{4}{c}{$\sqrt{n}(\tilde\theta_{EM}-\hat\theta^*_{EM})$} \bigstrut[t] \\
 \multicolumn{1}{r|}{$n:$} & 100 & 200 & 350 & 500 & 100 & 200 & 350 & 500 \\ 
   \hline
$\theta_{W,1}$ & 3.017 & 2.943 & 3.587 & 3.369 & 3.032 & 2.816 & 3.546 & 3.387 \\ 
  $\theta_{W,2}$ & 2.620 & 3.261 & 3.457 & 3.576 & 2.639 & 3.206 & 3.429 & 3.515 \\ 
  $\theta_{W,3}$ & 2.633 & 2.615 & 3.168 & 3.440 & 2.767 & 2.561 & 3.121 & 3.432 \\ 
  $\theta_{W,4}$ & 2.614 & 2.698 & 3.250 & 3.303 & 2.634 & 2.725 & 3.249 & 3.299 \\ 
  $\theta_{W,5}$ & 2.657 & 3.246 & 3.495 & 3.642 & 2.682 & 3.279 & 3.434 & 3.604 \\ 
  $\theta_{W,6}$ & 2.797 & 3.263 & 3.231 & 3.952 & 2.768 & 3.281 & 3.248 & 3.919 \\ 
  $\theta_{W,7}$ & 3.219 & 3.685 & 3.285 & 4.018 & 3.240 & 3.651 & 3.223 & 3.971 \\ 
  $\theta_{W,8}$ & 3.728 & 3.744 & 3.768 & 4.029 & 3.686 & 3.739 & 3.759 & 3.959 \\ 
  $\theta_{W,9}$ & 3.878 & 4.468 & 4.237 & 4.101 & 3.907 & 4.366 & 4.242 & 4.083 \\ 
  $\theta_{EC}$ & 0.554 & 0.376 & 0.349 & 0.335 & 0.279 & 0.223 & 0.170 & 0.129 \\ 
  $\theta_{FC}$ & 1.064 & 0.912 & 0.719 & 0.733 & 0.511 & 0.554 & 0.354 & 0.328 \\ 
   \hline
\end{tabular}
\caption{Empirical mean squared error of $\sqrt{n}(\tilde\theta_{H}-\hat\theta^*_{H})$ and $\sqrt{n}(\tilde\theta_{EM}-\hat\theta^*_{EM})$, for each sample size and element of $\theta$, calculated over the 100 replications.  $H$ denotes the estimator based upon \textcite{bajari2011simple} and $EM$ denotes the estimator based upon \textcite{arcidiacono2011conditional}.} 
\label{tab:conv2}
\end{table}

\subsection{Equation \eqref{eq:assumption1_section5} in Section \ref{sec:appl-model}}\label{appendix:monotoinicityinSection5}

In this section, we show \eqref{eq:assumption1_section5} under the following assumption on the initial condition:
$$
Pr(N_1=0\mid \lambda,W)=Pr(N_1=0\mid \lambda,\theta_W^\intercal W),
$$
$$
Pr(N_1=0\mid \lambda,\theta_W^\intercal W)>0,\mbox{ and }
$$
$$\frac{\partial}{\partial u}Pr(N_1=0\mid \lambda,\theta_W^\intercal W=u)\leq 0.
$$
Let $P_u(n)$ be the conditional choice probability of $A_{i,t}=1$ conditional on $\lambda_i + \theta^\intercal_WW_i=u$ and $N_{i,t}=n$.
The $\implies$ part follows from the following equalities:  
\begin{align*}
\Pi_0(w,t)
&=
Pr(N_t=0\mid W=w)
\\
&=
E[
Pr(N_t=0\mid \lambda,W)\mid W=w]
\\
&=
E[Pr(N_1=0\mid \lambda,W)\prod_{s=1}^{t-1}Pr(A_s=0\mid N_s=0,\lambda,W)\mid W=w]
\\
&=
E[Pr(N_1=0\mid \lambda,\theta_W^\intercal W=\theta_W^\intercal w)\prod_{s=1}^{t-1}Pr(A_s=0\mid N_s=0,\lambda,\theta_W^\intercal W=\theta_W^\intercal w)\mid W=w]
\\
&=
E[Pr(N_1=0\mid \lambda,\theta_W^\intercal W=\theta_W^\intercal w)\prod_{s=1}^{t-1}(1-P_{\theta_W^\intercal w}(0))\mid W=w]
\\
&=
E[Pr(N_1=0\mid \lambda,\theta_W^\intercal W=\theta_W^\intercal w)\prod_{s=1}^{t-1}(1-P_{\theta_W^\intercal w}(0))],
\end{align*}
where the third equality follows from the independence  between $A_t$ and $(A_{t-1},N_{t-1},\ldots,A_{1},N_{1})$ given $(N_t,\lambda,W)$, the fourth equality follows from the fact the initial distribution of $N_1$ and the utility function depends on $(\lambda,W)$ only through $(\lambda,\theta_W^\intercal W)$, and the last equality follows from the independence  between $\lambda$ and $W$.
To show the $\Longleftarrow$ part of \eqref{eq:assumption1_section5}, we show show the mapping 
$$
u\mapsto E[Pr(N_1=0\mid \lambda,\theta_W^\intercal W=u)\prod_{s=1}^{t-1}(1-P_{u}(0))]
$$
is strictly decreasing in $u$. 
Since $\frac{\partial}{\partial u}Pr(N_1=0\mid \lambda,\theta_W^\intercal W=u)\leq 0$, it is sufficient to show $\frac{d}{du}P_u(n)>0$.
The equilibrium ex-ante value function satisfies 
$$
{v}_u(n)=\sum_{t=0}^{\infty}\beta^tE[A_{i,s+t}(u
- \left(\theta_{FC}N_{i,s+t} + \theta_{EC} \mathsf{1}(N_{i,s+t}=0)\right)+e_{N_{i,s+t}}(P_u(N_{i,s+t})))\mid N_{i,s}=n].
$$
where $e_n(P_u(n))=E[\epsilon_{i,t}\mid A_{i,t}=1,N_{i,t}=n]$ and we omit the conditioning event $\lambda_i + \theta^\intercal_WW_i=u$ for notational simplicity. By the zero Jacobian property (cf. Proposition 2 of \cite{aguirregabiria2002swapping}), 
$$
\frac{d}{du}{v}_u(n)=\sum_{t=0}^{\infty}\beta^tE[A_{i,s+t}\mid N_{i,s}=n].
$$
This equation implies the recursive relationship 
$$
\frac{d}{du}{v}_u(n+1)=\frac{(1-\beta (1-P_u(n)))}{\beta P_u(n)}\frac{d}{du}{v}_u(n)-\frac{1}{\beta}
\mbox{ for }n<3,
$$
because 
\begin{align*}
\frac{d}{du}{v}_u(n)
&=
\sum_{t=0}^{\infty}\beta^tE[A_{i,s+t}\mid N_{i,s}=n]\\
&=
P_u(n)
+\sum_{t=1}^{\infty}\beta^t(P_u(n)E[A_{i,s+t}\mid N_{i,s+1}=n+1]
+(1-P_u(n))E[A_{i,s+t}\mid N_{i,s+1}=n])
\\
&=
P_u(n)+
\beta P_u(n)\frac{d}{du}{v}_u(n+1)
+\beta (1-P_u(n))\frac{d}{du}{v}_u(n).
\end{align*}
Then the choice-specific value function satisfies
$$
\frac{d}{du}(v_u(1,n)-v_u(0,n))
=
\begin{cases}
1+\beta \frac{d}{du}({v}_u(n+1)-{v}_u(n))
=
 \frac{1-\beta}{ P_u(n)}\frac{d}{du}{v}_u(n)
>0&\mbox{ if }n<3\\
1&\mbox{ if }n=3,
\end{cases}
$$
which implies $\frac{d}{du}P_u(n)=\frac{d}{du}F_{-\epsilon}(v_u(1,n)-v_u(0,n))>0$.
\end{document}